\documentclass[journal,12pt,onecolumn,draftclsnofoot]{IEEEtran}

\usepackage{graphicx}%
\usepackage{subcaption}
\usepackage{multirow}%
\usepackage{amsmath,amssymb,amsfonts}%
\usepackage{amsthm}%
\usepackage{mathrsfs}%
\usepackage{xcolor}%
\usepackage{textcomp}%
\usepackage{manyfoot}%
\usepackage{booktabs}%
\usepackage{algorithm}%
\usepackage{algorithmicx}%
\usepackage{algpseudocode}%
\usepackage{listings}%
\usepackage{makecell}
\usepackage{comment}

\theoremstyle{thmstyleone}%
\newtheorem{theorem}{Theorem}

\newtheorem{proposition}[theorem]{Proposition}%

\newtheorem{lemma}[theorem]{Lemma}

\theoremstyle{thmstyletwo}%

\theoremstyle{thmstylethree}%
\newtheorem{definition}{Definition}%

\begin{document}

\title{Probabilistic Risk Sensitivity and Loss Aversion in Cumulative Prospect Theory}

\author{Symeon~Vaidanis~and~Marios~Kountouris
\thanks{S. Vaidanis is with the Communication Systems Department of EURECOM, Sophia Antipolis, 06410 France, and M. Kountouris is with the Communication Systems Department of EURECOM, Sophia Antipolis, 06410 France and Department of Computer Science and Artificial Intelligence of University of Granada, C/ Periodista Daniel Saucedo Aranda, Granada, 18071 Spain. E-mail: symeon.vaidanis@eurecom.fr, marios.kountouris@eurecom.fr, mariosk@ugr.es.}}

\maketitle

\begin{abstract}
This paper develops a binary-gamble framework for characterizing risk sensitivity and loss aversion in Cumulative Prospect Theory (CPT). The proposed probabilistic risk-sensitivity metric is defined as a probability-threshold ratio that determines acceptance and preference thresholds in choice problems involving either a certain outcome and a binary gamble or two binary gambles. We show how standard notions of symmetric and non-symmetric bet aversion can be recovered within this framework, and we compare the resulting threshold-based conditions with utility premia, probability premia, and Arrow--Pratt curvature measures. The analysis clarifies when these criteria coincide and when they diverge, particularly for increasing aversion conditions, binary gambles with unequal probability distributions, and settings involving probability weighting functions. We also identify technical restrictions that arise when CPT-utility functions are used to represent loss aversion at the reference point. The resulting framework provides a decision-theoretic interpretation of risk sensitivity that is directly tied to probability thresholds and complements existing premium-based approaches.
\end{abstract}


\section{Introduction}\label{Section1}

Decision making is a fundamental and ubiquitous aspect of human life, manifesting in a wide range of contexts, from social interactions and everyday choices to high-stakes situations involving financial decisions or health-related judgments. Individuals are continually required to evaluate alternatives, assess potential outcomes, and navigate trade-offs between gains and losses, often under conditions of stochastic environments, subjective evaluation of decision criteria, and incomplete information. The study of such decision processes, shaped not only by objective probabilities but also by subjective perceptions and psychological factors, has long held a central place in economics, psychology, and behavioral sciences, motivating the development of models aimed at interpreting the complex nature of human choice.

Building on the foundational work of D. Bernoulli, J. von Neumann, O. Morgenstern, F. P. Ramsey, and L. J. Savage, among others, Expected Utility Theory (EUT) \cite{EUT} has profoundly shaped modern economic theory, statistical reasoning, and decision science by providing a normative framework for describing rational choice under a subjective perspective among alternatives in a stochastic environment. Despite its elegance and conceptual appeal, EUT has faced significant empirical challenges, as numerous studies have shown that individuals often deviate systematically from its predictions. 
In response, Prospect Theory (PT) was introduced by \cite{Kahneman_Tversky_1979} as a descriptive model that accounts for key behavioral patterns, offering a more psychologically grounded perspective on how people make decisions among alternatives in a stochastic environment. PT extends EUT by introducing a more general framework that accounts for both gains and losses relative to a reference point. More formally, PT partitions the domain of the utility function into gains and losses, assumes risk aversion in the gain domain and risk seeking in the loss domain, and applies a non-linear transformation to outcomes.  Moreover, PT introduces probability weighting functions (PWFs), which distort objective probabilities to reflect the subjective perception of event occurrence. These functions capture the empirically observed tendencies of individuals to overweight low-probability events and underweight high-probability ones, behaviors that traditional models fail to explain. 
To address certain limitations of the original formulation, particularly in handling choices involving ranked probabilities, Cumulative Prospect Theory (CPT) was later developed by \cite{Kahneman_Tversky_1992} and further refined through the contributions of \cite{Kobberling_Wakker}, \cite{Prelec} and \cite{Neilson_2002_A}. CPT extends PT by introducing cumulative probability weighting, which is a critical modification of the PT framework in order not to violate first-order stochastic dominance.
As a result, CPT has become a central model in behavioral decision theory, with widespread applications in economics, finance, and psychology for explaining systematic deviations from rational choice. Beyond these traditional domains, CPT has also been applied in areas such as reinforcement learning \cite{CPT_RL}, control theory \cite{CPT_Control}, and telecommunications \cite{CPT_Alouini}. More recently, CPT has been leveraged in the emerging field of goal-oriented semantic information theory \cite{kountouris2021semantics}, where concepts such as risk sensitivity and subjective outcome evaluation are used to characterize the semantic relevance and perceived value of information or entities of the telecommunication system, \cite{vaidanis2025ICC, vaidanis2025SPAWC}.

In this work, we revisit core components of CPT by developing a threshold-based account of risk sensitivity and loss aversion at the reference point. Inspired by the graphical representation of risk premia in \cite{utility_premium_fundamental}, we use binary gambles to define a probability-threshold metric of risk sensitivity. This metric, which we call probabilistic risk sensitivity, provides a structured way to describe subjective risk preferences beyond the standard premium-based formulations. By reformulating symmetric and non-symmetric bet aversion, together with their increasing versions, we extend the analysis to nested binary gambles with equal and unequal probability distributions. We also analyze the classical EUT (Bernoulli) utility function through this binary-gamble framework, clarify the relation with probability premia, and show how cumulative probability weighting changes the relation between them. 

Our analysis shows that binary gambles serve not merely as illustrative tools, but as a robust methodological foundation capable of recovering and extending core CPT constructs.  Taken together, these contributions enrich the theoretical foundations of CPT by introducing structurally grounded and behaviorally coherent extensions, advancing risk-sensitive decision theory and offering a model that better reflects the complexities of human judgment under uncertainty.

The remainder of the paper is organized as follows. Section~\ref{Section2} reviews the CPT framework and fixes notation. Section~\ref{Section3} introduces probabilistic risk sensitivity for binary gambles and relates it to standard definitions of risk sensitivity and loss aversion. Section~\ref{Section4} compares utility premia, probability premia, Arrow--Pratt measures, and probabilistic risk sensitivity under a mean-value decision criterion. Section~\ref{Section5} concludes.

\section{Preliminaries on Cumulative Prospect Theory} \label{Section2}

In this section, we provide an overview of the mathematical framework underlying Cumulative Prospect Theory (CPT) \cite{Kahneman_Tversky_1992,WakkerTve93,ChaWak99}. Our goal is to introduce the key components and foundations that define CPT, including utility functions, probability weighting functions, and decision weights. 
This exposition lays the theoretical foundation for the subsequent analysis in order to provide the reader with a clear and rigorous understanding of the fundamental components of the model.

CPT offers a compelling alternative and can be viewed as a generalization of EUT for modeling decision-making under risk. CPT provides a framework for evaluating choices among a finite set of prospects, represented as a collection of real-valued random variables $\mathcal{G}$. For each gamble $R \in \mathcal{G}$, the agent assigns a deterministic valuation $V(R)$ based on the CPT framework. The decision rule is to choose the prospect with the highest subjective valuation: $R^* = \arg\max_{R \in \mathcal{G}} V(R)$.

CPT introduces the following key features that depart from the assumptions of EUT: (i) \textit{reference dependence}, (ii) \textit{loss aversion}, and (iii) \textit{rank dependence and cumulative probability weighting}. Moreover, CPT incorporates the following two subjective or behavioral features:
\begin{itemize}
\item \textit{Diminishing sensitivity to outcomes}: sensitivity to changes in outcomes diminishes as the distance from the reference point increases, for both gains and losses. As we detail below, this implies that the utility function is concave for gains and convex for losses, reflecting decreasing marginal sensitivity in both domains. This property, grounded in a quantitative analysis of the utility (value) function, is consistent with the probabilistic (prospect-based) analysis that relies on Jensen's inequality. In the original work of Kahneman and Tversky \cite{Kahneman_Tversky_1979}, this probabilistic (prospect-based) analysis was used to justify the concavity within each subdomain.
\item \textit{Probabilistic sensitivity}: individuals do not perceive and respond to objective probabilities in a linear manner. Instead, they tend to overweight small probabilities and underweight large ones. This behavior is captured by the probability weighting function in CPT, which reflects systematic distortions in the perception of probability.
\end{itemize}
It is important to note that \textit{diminishing sensitivity} is conceptually related to, but distinct from, \textit{loss aversion}: while loss aversion refers to the asymmetric valuation of losses relative to gains, diminishing sensitivity pertains to the curvature of the utility function within each domain. However, in this work, we demonstrate that non-linear sensitivity constitutes the fundamental mechanism of CPT, since the principle of saturation effect of diminishing sensitivity is inconsistent with Neilson’s formulation of aversion, which definition will be related in following section with the non-symmetric binary gambles. On the other hand, probabilistic sensitivity is directly tied to the probability weighting function, highlighting the departure from the linear treatment of probabilities assumed in EUT.

A decision maker, or more generally, an agent, is characterized by a reference point $x_0 \in \mathbb{R}$, a corresponding utility function $u: \mathbb{R} \to \mathbb{R}$, and a pair of probability weighting functions $w^{+}$, $w^{-}: [0, 1] \to [0, 1]$, corresponding to gains and losses, respectively. These components jointly capture how the agent evaluates outcomes and perceives uncertainty. We refer to the triple $(x_0, u, w^{\pm})$ as the CPT features of the agent.

\subsection{Utility function and Reference Dependence}

In CPT, value is perceived and preferences are shaped by deviations from a specified \emph{reference point} $x_0$. This reference point represents a benchmark level, such as an acquired or expected outcome, status quo, operating level, or satisfaction level, against which gains and losses are evaluated. The choice of $x_0$ may vary across individuals, contexts, and application scenarios, depending on psychological or situational factors.

A CPT-utility function is a piecewise-defined mapping $u:\mathbb{R} \to \mathbb{R}$  given by:
\[
u(x) =
\begin{cases}
u^+(x - x_0), & \text{if } x \geq x_0, \\
- u^-(x_0 - x), & \text{if } x < x_0,
\end{cases}
\]
where \( u^+: \mathbb{R}_{\geq 0} \to \mathbb{R}_{\geq 0} \) is the gain component, \( u^-: \mathbb{R}_{> 0} \to \mathbb{R}_{\geq 0} \) is the loss component, both \( u^+ \) and \( u^- \) are strictly increasing, continuous, and \( u^-(x) \geq u^+(x) \) for all \( x > 0 \), modeling loss aversion. In CPT, the utility function is typically concave over gains and convex over losses, reflecting diminishing sensitivity in both domains (with risk aversion for gains and risk seeking for losses around the reference point).

The domain of the value (utility) function is partitioned relative to $x_0$ into two regions: the loss domain $x < x_0$ and the gain domain $x \geq x_0$. In the loss domain, outcomes are assigned negative value, i.e., $u(x) < 0, \forall x < x_0$, and typically $\lim _{x \to x_0^{-}}u(x) \leq 0$. In the gain domain, outcomes are valued positively, i.e., $u(x)>0, \forall x > x_0$, and typically $\lim _{x \to x_0^{+}}u(x) \geq 0$. 
In general, the value at the reference point $u(x_0)$ can lie anywhere in the interval $\lim _{x \to x_0^{-}}u(x) \leq u(x_0) \leq \lim _{x \to x_0^{+}}u(x)$, exhibiting a discontinuity at $x_0$, i.e., $\lim_{x \to x_0^-} u(x) \neq \lim_{x \to x_0^+} u(x)$. This could make the utility function $u$ non-differentiable or non-smooth at $x_0$. However, in standard formulations of CPT, it is customary to normalize the reference point such that $\lim _{x \to x_0^{-}}u(x) = \lim _{x \to x_0^{+}}u(x) = u(x_0) = 0$, which simplifies analysis and reflects the idea that the reference point represents a neutral benchmark, neither a gain nor a loss.

This reference-dependent evaluation is a key departure from EUT, where utility is typically defined over absolute outcomes and rationality. In contrast, CPT’s distinction between gains and losses introduces systematic asymmetries in valuation, which are essential for capturing a wide range of empirically observed, even irrational, decision-making behaviors under risk.

\subsection{Risk Sensitivity}

In the context of monotonicity, the utility function is assumed to be strictly increasing over its entire domain, i.e., $u(x_1) < u(x_2), \; \forall x_1 < x_2$. As a result, the first derivative of the utility function satisfies:
\begin{equation}
\frac{\partial{u}}{\partial{x}} > 0, \quad \forall \; x \in \mathbb{R} \setminus \{x_0\}.
\end{equation}
The curvature of the utility function (marginal utility), reflected in its second derivative, characterizes the agent’s attitude toward risk (risk sensitivity) and captures how the agent perceives changes of varying magnitudes within each subdomain (gains or losses).
In the gain domain ($x > x_0$):
\begin{itemize}
    \item A convex utility function indicates increasing marginal utility, meaning the agent perceives larger gains (further from $x_0$) as increasingly valuable.
    \item A concave utility function indicates diminishing marginal utility, meaning the agent is more sensitive to changes (gains) closer to $x_0$.   
\end{itemize} 
In the loss domain ($x < x_0$), the interpretation is reversed:
\begin{itemize}
    \item A convex utility function reflects diminishing sensitivity to larger losses (farther from $x_0$),
    \item whereas a concave function implies increasing sensitivity to such losses.
\end{itemize}
This diverges from the standard assumption in EUT, where the utility function is typically assumed to be globally concave.

Furthermore, we consider a decision-maker choosing between an $N$-outcome gamble $X$ with probability distribution $\mathcal{P}$ and a sure (stationary) outcome $x_1$, under fairness condition $\mathbb{E}_P(X) = x_1$. The outcomes of the gamble lie within a single subdomain (i.e., entirely in the gain domain or entirely in the loss domain). In this setting:
\begin{itemize}
    \item A concave utility function represents a \emph{risk-averse} attitude, since the agent prefers the sure (stationary) outcome to the gamble; by Jensen's inequality this is expressed as $\mathbb{E}_P(u(X))< u(x_1) = u(\mathbb{E}_P(X))$.
    \item A convex utility function represents \emph{risk-seeking} behavior, since the agent prefers the gamble to the sure outcome; by Jensen's inequality this is expressed as $\mathbb{E}_P(u(X)) > u(x_1) = u(\mathbb{E}_P(X))$.
\end{itemize}
If the agent’s preferences are invariant to the scale of outcomes, their utility function is linear, indicating risk neutrality. In contrast, a constant utility function models an agent who is completely insensitive to changes in outcomes, expressing indifference across all values.
Within the classical CPT framework, the utility function is typically assumed to be S-shaped, satisfying:
\begin{align*}
\frac{\partial^2 u}{\partial x^2} &\leq 0 \quad \forall, x > x_0 \quad &\text{(concave in the gain domain)}, \\
\frac{\partial^2 u}{\partial x^2} &\geq 0 \quad \forall, x < x_0 \quad &\text{(convex in the loss domain)}.
\end{align*}
This structure captures the principle of diminishing marginal sensitivity: agents are more sensitive to changes close to $x_0$ and less sensitive as outcomes move further into the gain or loss regions.

\subsection{Arrow-Pratt Analysis about Risk Sensitivity}

In the framework of risk analysis developed by Arrow and Pratt \cite{Arrow,Pratt}, the main tools are the notions of the certainty equivalent and the risk premium. The certainty equivalent of a gamble $X$ with probability distribution $P$ is defined as $x_{CE} = u^{-1}\bigl(\mathbb{E}_P[u(X)]\bigr)$, and can be interpreted as the sure amount that delivers the same level of utility as the (expected) utility of the gamble under the given utility function. The risk premium can be divided into two subcases: the absolute risk premium and the relative risk premium, which are defined as follows:
\begin{subequations}
    \begin{equation}
         \pi_A = \mathbb{E}_P(X) - x_{CE}
    \end{equation}
    \begin{equation}
        \pi_R = \frac{\pi_A}{\mathbb{E}_P(X)} = 1 - \frac{x_{CE}}{\mathbb{E}_P(X)}
    \end{equation}
\end{subequations}
Here, $\pi_A$ denotes the absolute risk premium and $\pi_R$ the relative risk premium, both expressed in terms of the certainty equivalent $x_{CE}$ and the expected value $\mathbb{E}_P(X)$ of the gamble.

In their original papers, Arrow and Pratt first apply a first-order Taylor expansion of $u$ around the certainty equivalent and a second-order Taylor expansion around the mean of the gamble. They then take expectations of the second-order expansion of $u(x)$ around the mean. In this way, the absolute and relative risk premia are obtained in terms of the local curvature of the utility function as follows:
\begin{subequations}
    \begin{equation}
        \pi_A \approx - \frac{1}{2} \cdot \frac{u''(\mathbb{E}_P(X))}{u'(\mathbb{E}_P(X))} \cdot \sigma_X^2
        = \frac{1}{2} A(\mathbb{E}_P(X)) \cdot \sigma_X^2
        \;,\;
        A(x) = - \frac{u''(x)}{u'(x)}
    \end{equation}
    \begin{equation}
        \pi_R \approx - \frac{1}{2} \cdot \frac{u''(\mathbb{E}_P(X)) \cdot \mathbb{E}_P(X)}{u'(\mathbb{E}_P(X))} \cdot \frac{\sigma_X^2}{\mathbb{E}_P^2(X)}
        = \frac{1}{2} R(\mathbb{E}_P(X)) \cdot \frac{\sigma_X^2}{\mathbb{E}_P^2(X)}
        \;,\;
        R(x) = - \frac{x\,u''(x)}{u'(x)}
    \end{equation}
\end{subequations}
where $\sigma_X^2$ is the approximation of variance of the gamble, and $A(x)$ and $R(x)$ are the Arrow-Pratt coefficients of absolute (ARA) and relative risk aversion (RRA), respectively. The function, which has constant absolute risk aversion (CARA), is any affine transformation of the exponential function, $u(x) = \mu + \kappa \cdot \exp{(\alpha \cdot x)}$ and $A(x) = -\alpha$.

In the case of fair decision setup between a stationary state $x$ and a symmetric binary gamble around it $\{x-\delta,\frac{1}{2};x+\delta,\frac{1}{2}\}$, the absolute risk premium is defined as following:
\begin{equation}
    u(x+\pi_A(x,\delta,u)) = \frac{1}{2}u(x+\delta) + \frac{1}{2}u(x-\delta)
\end{equation}
and regarding the works of \cite{utility_premium_fundamental}, the utility premium for the previous defined decision setup is defined as
\begin{equation}
    \mathcal{U}(x,\delta) = u(x) - \left( \frac{1}{2}u(x+\delta) + \frac{1}{2}u(x-\delta) \right)
\end{equation}
We have to notice that the component $\mathbb{E}(x,\delta) = \frac{1}{2}u(x+\delta) + \frac{1}{2}u(x-\delta)$ is the mean value of the gamble.
Moreover, in a case of a decision setup between a stationary state $x$ and a symmetric binary gamble around it $\{x-\delta,\frac{1}{2};x+\delta,\frac{1}{2}\}$, the probability premium $\gamma$ is the excess in winning probability over a fair condition that makes the agent indifferent between the stationary state and the binary gamble, \cite{microeconomic_theory}.
\begin{equation}
    u(x) = \left(\frac{1}{2} + \gamma(x,\delta)\right) u(x+\delta) + \left(\frac{1}{2} - \gamma(x,\delta)\right)u(x-\delta)
\end{equation}

\subsection{Loss Aversion}

The CPT utility function exhibits \textit{loss aversion}, which refers to the phenomenon by which agents are more sensitive to losses than to equivalent gains. The utility function is steeper for losses than for gains, reflecting a behavioral bias in which individuals (agents) exhibit a stronger sensitivity to losses. 

Kahneman and Tversky introduced a related behavioral condition known as \emph{symmetric bet aversion} \cite{Kahneman_Tversky_1979}. This condition captures the idea that an agent strictly prefers not to engage in a symmetric gamble that offers equal magnitude gains and losses around the reference point. It is formally defined as:
\begin{equation} \label{eq: KT_loss_aversion}
    u(x_0 + \delta) + u(x_0 -\delta) < 0 \;,\; \forall \delta>0 \;,\; u(x_0)=0.
\end{equation} 
This condition implies that all symmetric fair gambles around the reference point are strictly rejected in favor of maintaining the status quo, reflecting a preference for certainty over balanced risk. By dividing each side with $\delta$ and taking the limit $\delta \to 0^+$, the marginal utility for losses exceeds that for gains at the reference point. Formally, this is expressed as: $u'(x_0^+) < u'(x_0^-)$, where $u'(x_0^+)$ and $u'(x_0^-)$ denote the right- and left-hand derivatives of the utility function at the reference point $x_0$, respectively.

They also introduced a stronger version of this condition, known as \emph{increasing symmetric bet aversion}, which imposes an ordering on the aversion to symmetric bets of increasing magnitude. It is expressed as: 
\begin{equation} \label{eq:KT_increasing_loss_aversion_1}
    u(x_0 + \delta_1) + u(x_0 - \delta_1) < u(x_0 + \delta_2) + u(x_0 - \delta_2) \;,\; \forall 0 \leq \delta_2 < \delta_1.
\end{equation} 
This condition states that the disutility of a symmetric gamble becomes more pronounced as the stakes increase. Under suitable regularity assumptions (e.g., differentiability of $u$), it can be reformulated in differential terms as:
\begin{equation} \label{eq:KT_increasing_loss_aversion_derivative}
    \frac{\partial{u}}{\partial{x}} \bigg\vert _{x = x_0 + \delta} < \frac{\partial{u}}{\partial{x}} \bigg\vert _{x = x_0 - \delta} \;,\; \forall \delta>0 \;,\; u(x_0) = 0.
\end{equation}
This inequality reinforces the asymmetry between gains and losses, emphasizing that marginal sensitivity to losses consistently exceeds that to gains for all deviations from the reference point.
Moreover, this condition implies that the degree of aversion to symmetric fair gambles increases monotonically with the magnitude of the deviation $\delta$, i.e., the likelihood of rejection becomes stronger as the absolute size of the gamble increases.

\cite{Neilson_2002_A} extended the standard definitions by introducing two refined concepts: \emph{Neilson's weak aversion} and \emph{Neilson's strong aversion}.
Weak loss aversion is defined by the condition
\begin{equation} \label{eq:Neilson_weak}
    \frac{u(z)}{z-x_0} \leq \frac{u(y)}{y-x_0} \;,\; \forall y < x_0 < z,
\end{equation}
while strong loss aversion strengthens this notion by comparing local marginal utilities:
\begin{equation} \label{eq:Neilson_strong}
    \frac{\partial{u}}{\partial{x}} \bigg\vert _{x = z} \leq \frac{\partial{u}}{\partial{x}} \bigg\vert _{x = y} \;,\; \forall y < x_0 < z.
\end{equation} 
Neilson's weak aversion implies the weak rejection of all fair non-symmetric binary gambles around the reference point in favor of maintaining the status quo. Neilson's strong aversion implies Neilson's weak aversion. The converse requires additional regularity and tail conditions; in particular, for nonsmooth S-shaped utility functions, weak aversion requires a suitable no-saturation or tail-slope condition on the loss domain. Under such a condition, the equivalence between Neilson’s two loss-aversion definitions can be established. Most loss-aversion definitions in the literature, with the exception of Neilson's strong aversion, can be interpreted through the utility-premium approach. More precisely, for the fair non-symmetric binary gamble with outcomes $x_0-\delta_1$ and $x_0+\delta_2$, define $\mathbb{E}(x_0,\delta_1,\delta_2)=\frac{\delta_1}{\delta_1+\delta_2}u(x_0+\delta_2)+\frac{\delta_2}{\delta_1+\delta_2}u(x_0-\delta_1)$ and $\mathcal{U}(x_0,\delta_1,\delta_2)=u(x_0)-\mathbb{E}(x_0,\delta_1,\delta_2)$. The relevant conditions are:
\begin{itemize}
    \item Symmetric bet aversion: $ \mathcal{U}(x_0,\delta,\delta) > 0 \Leftrightarrow \mathbb{E}(x_0,\delta,\delta) < u(x_0)$,
    \item Increasing symmetric bet aversion: $ \frac{\partial \mathcal{U}(x_0,\delta,\delta)}{\partial \delta} > 0 \Leftrightarrow\frac{\partial \mathbb{E}(x_0,\delta,\delta) }{\partial \delta}< 0$,
    \item Neilson's weak aversion: $\mathcal{U}(x_0,\delta_1,\delta_2) > 0 \Leftrightarrow \mathbb{E}(x_0,\delta_1,\delta_2) < u(x_0)$.
\end{itemize}

\subsection{Probability Weighting Functions}

A key feature of CPT is the non-linear distortion of probabilities, where objective probabilities are transformed through a PWF to reflect how individuals subjectively perceive uncertainty. The PWF captures the empirically observed tendency for individuals to overweight small probabilities and underweight moderate to high probabilities, deviating from the linear treatment in EUT. 

Formally, the probability weighting functions are denoted by $\omega^{\pm}: [0, 1] \to [0, 1]$, where $\omega^{+}$ is applied to gains and $w^{-}$ to losses. These functions are continuous, strictly increasing, and satisfy the boundary conditions: $\omega^{\pm}(0) = 0$ and $\omega^{\pm}(1) = 1$. 
It is important to note that the functions $\omega^{+}$ and $\omega^{-}$ assign the same decision weight to the reference point outcome if and only if the following symmetry condition holds: $\omega^-(p) + \omega^+(1-p) = 1 \; \forall p \in [0, 1]$, which is equivalent to using a single, unified PWF for the entire cumulative distribution \cite{Ingersoll}. 
Moreover, a fundamental distinction between PT and CPT lies in the domain over which the PWFs are applied. In PT, the weighting is applied directly to the probability mass function (in discrete settings) or probability density function (in continuous settings). This approach can violate first-order stochastic dominance. By contrast, CPT applies the weighting to cumulative probabilities, i.e., the cumulative distribution function (CDF) in the continuous case, preserving consistency with stochastic dominance and enabling a rank-dependent evaluation of outcomes.

One of the earliest and most influential forms of the PWF was proposed in \cite{Kahneman_Tversky_1992}, based on empirical and experimental evidence, and is defined as:
\begin{equation}
    \omega(p) = \frac{p^\delta}{\left( p^\delta + (1-p)^\delta \right)^{1/\delta}} \;,\; 0 < \delta \leq 1.
\end{equation}

The log-odds distortion function $w_{p_0,\gamma}$ is suggested in \cite{LogOdd95}, which is defined as:
\begin{equation}
    \text{Lo}(\omega_{p_0,\gamma}(p)) = \text{Lo}(p) + (1 - \gamma)\, \text{Lo}(p_0), \quad \forall\, p \in (0, 1),
\end{equation}
where $\text{Lo}(p) := \ln\left( \frac{p}{1 - p} \right)$ denotes the log-odds transformation.

One of the most widely used forms of the probability weighting function is the Prelec function \cite{Prelec}, defined as:
\begin{equation}
    \omega(p) = \exp \left( -\gamma \left( -\ln(p) \right)^\theta \right)
    \label{eq:Prelec_PWF}
\end{equation}
where $p \in (0,1)$, $\theta > 0$, and $\gamma > 0$. 
The parameter $\theta$ controls the shape or curvature of the function, shaping the degree of probability distortion, while $\gamma$ determines the overall elevation of the curve, i.e., the location of the inflection point relative to the identity line $\omega(p)=p$, effectively shifting the function upward or downward.  For $0 < \theta < 1$ the function exhibits the familiar inverse-S pattern, with overweighting of small probabilities and underweighting of large probabilities. When $\theta > 1$, the weighting function becomes S-shaped instead. Table~\ref{tab:parameters_probability_weighting_function} summarizes the main shape properties for the parameter combinations considered in this paper.
\begin{table}[h]
    \centering
    \begin{tabular}{|c|c|c|c|}
        \hline
         & $0 < \gamma < 1$ & $\gamma = 1$ & $1 < \gamma$\\
        \hline
        $0 < \theta < 1$ & \makecell{inverse S-shape, \\ $\tilde{p} < \omega(\tilde{p})$}  & \makecell{inverse S-shape, \\ $\tilde{p} = \omega(\tilde{p})$} & \makecell{inverse S-shape, \\ $\tilde{p} > \omega(\tilde{p})$} \\
        \hline
        $\theta = 1$ & \makecell{strictly concave, \\ $\tilde{p} < \omega(\tilde{p})$} & \makecell{linear, \\ $\tilde{p} = \omega(\tilde{p})$} & \makecell{strictly convex, \\ $\tilde{p} > \omega(\tilde{p})$} \\
        \hline
        $\theta > 1$ & \makecell{S-shape, \\ $\tilde{p} < \omega(\tilde{p})$} & \makecell{S-shape, \\ $\tilde{p} = \omega(\tilde{p})$} & \makecell{S-shape, \\ $\tilde{p} > \omega(\tilde{p})$} \\
        \hline
    \end{tabular}
    \caption{Shape properties of the Prelec probability weighting function. The inflection point is denoted by $\tilde{p}$.}
    \label{tab:parameters_probability_weighting_function}
\end{table}

Taken together, the parameter combinations in Table~\ref{tab:parameters_probability_weighting_function} produce both inverse-S and S-shaped weighting patterns. The standard CPT pattern of overweighting small probabilities and underweighting large probabilities arises in the inverse-S case; other parameter regimes may generate different distortions relative to the identity line.

\subsection{Common Parametric CPT Utility Functions}

Two widely used utility functions within the CPT framework have been proposed in the literature. The first is the Kahneman--Tversky utility function \cite{Kahneman_Tversky_1992}, defined as follows:
\begin{equation}
    u(x) = \left\{
        \begin{array}{ll}
            \left(x - x_0 \right)^\alpha & \textrm{for} ~x \geq x_0 \\
            - \lambda\left(x_0 - x \right)^\beta & \textrm{for} ~x < x_0 \\
        \end{array} 
    \right. 
    \label{eq:KT_utility}
\end{equation}
where $\alpha,\beta\in(0,1]$ govern curvature and capture diminishing sensitivity to gains and losses, respectively, while $\lambda>0$ represents the degree of loss aversion. Despite its intuitive appeal and empirical relevance, this formulation has two notable limitations: (i) it is generally not differentiable at $x_0$; for $0<\alpha,\beta<1$, both one-sided marginal values diverge to $+\infty$, whereas for $\alpha=\beta=1$ the one-sided derivatives are finite and equal to $1$ and $\lambda$, respectively; and (ii) symmetric bet aversion is satisfied only under the parameter restrictions identified in \cite{Al-Nowaihi}, including the benchmark case $\alpha=\beta$ and $\lambda>1$.

Another widely used utility function, proposed by K\"{o}bberling and Wakker \cite{Kobberling_Wakker}, is given by:
\begin{equation}
    u(x) = \left\{
        \begin{array}{ll}
            \frac{1 - \exp\left( - \alpha (x - x_0)  \right)}{\alpha} & \text{for}~ x_0 \leq x \\
            - \lambda \frac{1 - \exp\left( - \beta (x_0 - x)  \right)}{\beta} & \text{for}~ x < x_0 \\
        \end{array} 
    \right. 
    \label{eq:KW_utility}
\end{equation}
where $\lambda$, $\alpha$, and $\beta$ are agent-specific parameters. The positivity of all parameters ensures that the function is strictly increasing and exhibits the S-shape characteristic of CPT.
This formulation allows for greater flexibility than the classical power function, particularly in controlling curvature and asymmetry. Moreover, the condition for increasing symmetric bet aversion can be satisfied if $\alpha > \beta$ and $\lambda > 1$ \cite{Dhami_book}.

\subsection{Generic Framework of Cumulative Prospect Theory}

The subjective value attributed to the prospect $R$, which is a discrete random variable, is given by
\begin{equation}
    V(R) = \sum_{i=-n}^{m} \Omega_i \cdot u(r_i),
\end{equation}
where the cumulative probability weighting under the ordering $r_{-n} \leq r_{-n+1} \leq \dots r_{0} \leq \dots \leq r_{m-1} \leq r_{m}$ is the following
\begin{subequations}
    \begin{equation}
        P_{-i} = p(R = r_{-n}) + \dots + p(R = r_{-i}) \rightarrow \Omega_{-i} = \omega^-(P_{-i}) - \omega^-(P_{-i-1})
    \end{equation}
    \begin{equation}
        P_{i}^c = p(R = r_{i}) + \dots + p(R = r_{m}) \rightarrow \Omega_{i} = \omega^+(P_{i}^c) - \omega^+(P_{i+1}^c)
    \end{equation}
\end{subequations}
Following the equal-weighting argument for the reference-point outcome in \cite{Ingersoll}, a unified weighting representation is admissible only under the symmetry condition $\omega^-(p)+\omega^+(1-p)=1$ for all $p\in[0,1]$. This condition relates the gain and loss weighting functions but does not generally imply the stronger requirement $\omega^- = \omega^+$.

For a continuous prospect $R$, the CPT value may be written in Stieltjes form as
\begin{equation}
    V(R)=\int_{-\infty}^{x_0}u(r)\,d[\omega^-(F_R(r))]-\int_{x_0}^{+\infty}u(r)\,d[\omega^+(\bar F_R(r))].
\end{equation}
The minus sign in the gain term accounts for the fact that the survival function $\bar F_R(r)=1-F_R(r)$ is decreasing in $r$. Here $F_R(r)=\mathbb{P}(R\le r)$ and $\bar F_R(r)=\mathbb{P}(R>r)$ denote the CDF and survival function, respectively. Under a valid single-weighting representation, the valuation can be written as
\begin{equation}
    V(R)=\int_{-\infty}^{+\infty}u(r)\,f_R(r)\,\frac{d\omega(F_R(r))}{dF_R(r)}\,dr,
\end{equation}
where $\omega$ is the common weighting function in the unified representation and $f_R$ is the density of $R$.

\section{Probabilistic Risk Sensitivity: A Gamble-based Approach}\label{Section3}

The curvature of the utility function within each subdomain and the definition of loss aversion in CPT are traditionally motivated by the analysis of gambles. Arrow and Pratt \cite{Arrow,Pratt} characterize risk attitudes through risk and probability premia, whereas many loss-aversion conditions are formulated through the utility premium associated with gamble rejection. Here, we develop a complementary binary-gamble framework that interprets risk behavior through probability thresholds. The construction is inspired by the graphical approach of \cite{utility_premium_fundamental}; related graphical interpretations of probability premia appear in \cite{probability_premium_graphical}, but they do not directly provide the threshold-based comparison developed below.

\subsection{Basic Concept of Probabilistic Risk Sensitivity: Binary Gambles around a Stationary State}

To start with, we consider a binary gamble in which an agent faces two options: either to remain in the current state $x$, or to engage in a gamble that leads to one of two possible outcomes: a loss to state $(x - \delta_1)$ with probability $1-\rho$ or a gain to state $(x + \delta_2)$ with probability $\rho$, where $\delta_1, \delta_2 > 0$. These outcomes may lie in the same or in different subdomains relative to $x$. Importantly, the outcomes are not assumed to be symmetric.

Assuming the utility function $u(\cdot)$ is strictly increasing (i.e., strictly monotonic), we express the utility of the current state as a weighted average of the utilities of the potential outcomes of the gamble. That is,
\begin{equation}
u(x) = \tau \cdot u(x + \delta_2) + (1 - \tau) \cdot u(x - \delta_1), \quad \tau \in [0, 1].
\end{equation}
Here, $\tau$ represents the weight (or subjective probability) assigned to the gain outcome, and $(1 - \tau)$ is the corresponding weight for the loss outcome. This formulation reflects the indifference condition under which the decision-maker is just willing to accept the gamble rather than remain at state $x$.

The expected utility of the gamble is given by:
\begin{equation}
\mathbb{E}[u(\text{gamble})] = \rho \cdot u(x + \delta_2) + (1 - \rho) \cdot u(x - \delta_1), \quad \rho \in [0, 1].
\end{equation}
Now, consider the graphical representation of the utility function. If we draw a straight line $(\epsilon)$ connecting the points $u(x - \delta_1)$ and $u(x + \delta_2)$ on the utility curve, then every point on this chord represents a convex combination of the two endpoint utilities, and in particular the point generated by the weights $(1-\rho, \rho)$ has ordinate equal to the expected utility of the gamble. This statement can be extracted by the following arguments: The straight line $(\epsilon_1)$ is given by $f(\rho) = (1 - \rho)\,(x - \delta_1) + \rho\,(x + \delta_2),$ which lies on the $x$-axis, and the point $(f(\rho),0)$ describes the expected value of the gamble under absence of risk sensitivity, i.e., $f(\rho)\equiv \mathbb{E}[\text{gamble}]$. The ratio between length of the straight lines $\left((x - \delta_1,0),(f(\rho),0)\right)$ and $\left((f(\rho),0),(x + \delta_2,0)\right)$ is equal to $\frac{(x + \delta_2) - f(\rho)}{f(\rho) - (x - \delta_1)} = \frac{1-\rho}{\rho}.$ Continuing, by using fundamental geometric properties the following triangles are similar: $\left( (x - \delta_1,u(x - \delta_1)),(f(\rho),u(x - \delta_1)),(f(\rho),y_{(\epsilon)}) \right)$ and $\left( (x + \delta_2,u(x + \delta_2)),((x + \delta_2),y_{(\epsilon)}),(f(\rho),y_{(\epsilon)}) \right)$, where the $y_{(\epsilon)}$ is the co-domain value of the $(\epsilon)$ straight line at domain point $(f(\rho),0)$. The ratio between length of the straight lines $\left( (0,u(x + \delta_2)),(0,y_{(\epsilon)}) \right)$ and $\left( (0,y_{(\epsilon)}),(0,u(x - \delta_1)) \right)$ is equal to $\frac{u(x + \delta_2) - y_{(\epsilon)}}{y_{(\epsilon)} - u(x - \delta_1)} = \frac{1-\rho}{\rho}$, and hence $y_{(\epsilon)} = (1-\rho)\,u(x - \delta_1) + \rho\,u(x + \delta_2) \equiv \mathbb{E}[u(\text{gamble})]$. Therefore, the expected utility can be read directly as the ordinate of the point on the chord $(\epsilon)$ determined by the same probability weight $\rho$. Under the fairness condition $\mathbb{E}[\text{gamble}] = x$, we have $x = (1-\rho)(x-\delta_1)+\rho(x+\delta_2)$, which implies $\rho=\frac{\delta_1}{\delta_1+\delta_2}$. Hence, the point on the chord corresponding to the current state $x$ is obtained by taking $\rho=\frac{\delta_1}{\delta_1+\delta_2}$. In this line, we can identify the point on this line that corresponds to the current state $x$. This point, denoted by $A$, lies on the chord between the endpoints and is given by the convex combination:

\begin{equation}
A = \frac{\delta_1}{\delta_1 + \delta_2} \cdot u(x + \delta_2) + \frac{\delta_2}{\delta_1 + \delta_2} \cdot u(x - \delta_1)
\end{equation}
This point $A$ represents the linear interpolation of the utility values at the two outcomes, weighted according to their relative distances from the current state. It is useful for comparing the actual utility $u(x)$ with the linear approximation to assess the curvature of the utility function (i.e., risk aversion, neutrality, or seeking behavior).

It is important to note that the function
\begin{equation}
f(z) = z \cdot u(x + \delta_2) + (1 - z) \cdot u(x - \delta_1), \quad z \in [0, 1]
\end{equation}
is strictly increasing in $z$, as it is a convex combination of strictly increasing functions. This property allows us to compare the expected utility of the gamble to the utility of the current state $x$ in a meaningful way.

Under the expected utility criterion, the decision to reject the gamble occurs when: 
\[\mathbb{E}[u(gamble)] < u(x) \Rightarrow \rho < \tau.\]
Conversely, the decision to accept the gamble is made when:
\[\mathbb{E}[u(gamble)] > u(x) \Rightarrow \tau < \rho.\]
This implies that the parameter $\tau$ represents the minimum probability of gain required for the agent to be indifferent between playing the gamble and staying at $x$. Hence, we define $\tau \equiv \mathbb{P}_{gain}(x,\delta_1,\delta_2)$ and naturally, $\mathbb{P}_{gain}(x,\delta_1,\delta_2) + \mathbb{P}_{loss}(x,\delta_1,\delta_2) = 1$, which is the minimum level of certainty regarding the gain outcome required to justify playing the game.

Thus, the utility-based indifference condition becomes:
\[
    u(x) = \mathbb{P}_{gain}(x,\delta_1,\delta_2) \cdot u(x+\delta_2) + (1 - \mathbb{P}_{gain}(x,\delta_1,\delta_2)) \cdot u(x - \delta_1). 
\]
Solving for the gain threshold probability, we obtain:
\begin{equation}
    \mathbb{P}_{gain}(x,\delta_1,\delta_2) = \frac{u(x) - u(x - \delta_1)}{u(x+\delta_2) - u(x - \delta_1)}.
\end{equation} 
Equivalently, the ratio of the required gain probability to the complementary loss probability is:
\begin{equation}
    \frac{\mathbb{P}_{gain}(x,\delta_1,\delta_2)}{\mathbb{P}_{loss}(x,\delta_1,\delta_2)} =  \frac{u(x) - u(x - \delta_1)}{u(x+\delta_2) - u(x)}.
\end{equation} 
It is evident that a higher value of $\frac{u(x) - u(x - \delta_1)}{u(x + \delta_2) - u(x - \delta_1)}$ corresponds to a higher threshold probability of gain required to accept the gamble over the stationary status quo. Hence, this ratio expresses how much more utility is needed to compensate for the loss, relative to the potential gain. Therefore, it can serve as a local measure of risk sensitivity around the status quo, which we define as:
\begin{equation}
    \mathcal{RS}_{stat}(x,\delta_1,\delta_2) := \frac{u(x) - u(x - \delta_1)}{u(x+\delta_2) - u(x)}.
\end{equation}
Next, we define the following notions of risk sensitivity based on the mean value decision criterion:
\begin{itemize}
    \item \emph{Probabilistic symmetric bet aversion}: For a decision problem between a stationary state $x$ and a binary gamble $\{x-\delta,1-p;x+\delta,p\},\delta>0$, the subjective probability threshold described by a utility function $u$ is greater than the objective probability threshold described by a linear utility function
    \begin{equation}
        \mathcal{RS}_{stat}(x,\delta) > 1
    \end{equation}
    \item \emph{Probabilistic increasing bet aversion}: For a decision problem between a stationary state $x$ and a binary gamble $\{x-\delta,1-p;x+\delta,p\},\delta>0$, the subjective probability threshold described by a utility function $u$ is greater than the objective probability threshold described by a linear utility function and also the subjective probability threshold is an increasing function of $\delta$
    \begin{equation}
        \mathcal{RS}_{stat}(x,\delta) > 1 \text{ and } \frac{\partial \mathcal{RS}_{stat}(x,\delta)}{\partial \delta} > 0
    \end{equation}
    \item \emph{Probabilistic non-symmetric bet aversion}: For a decision problem between a stationary state $x$ and a binary gamble $\{x-w\cdot\delta_1,1-p;x+w\cdot\delta_2,p\},\delta_1,\delta_2,w>0$ for a fixed direction $(\delta_1,\delta_2)\in\mathbb{R}_{++}^2$ and stake scale $w>0$, the subjective probability threshold described by a utility function $u$ is greater than the objective probability threshold described by a linear utility function
    \begin{equation}
        \mathcal{RS}_{stat}(x,\delta_1,\delta_2,w) > \frac{\delta_1}{\delta_2}
    \end{equation}
    \item \emph{Probabilistic increasing non-symmetric bet aversion}: For a decision problem between a stationary state $x$ and a binary gamble $\{x-w\cdot\delta_1,1-p;x+w\cdot\delta_2,p\},\delta_1,\delta_2,w>0$ for a fixed direction $(\delta_1,\delta_2)\in\mathbb{R}_{++}^2$ and stake scale $w>0$, the subjective probability threshold described by a utility function $u$ is greater than the objective probability threshold described by a linear utility function and also the subjective probability threshold is an increasing function of $w$
    \begin{equation}
        \mathcal{RS}_{stat}(x,\delta_1,\delta_2,w) > \frac{\delta_1}{\delta_2} \text{ and } \frac{\partial \mathcal{RS}_{stat}(x,\delta_1,\delta_2,w)}{\partial w} > 0
    \end{equation}
\end{itemize}
We have to underline that the notion \emph{increasing} can have different interpretation. More precisely, in this subsection, we use this adjective definition in the context of a decision problem between a stationary state and a binary gamble. At the following subsection, we are going to discuss the approach of \emph{increasing} notion by Kahneman-Tversky as it is described in Equation \ref{eq: KT_loss_aversion}.

\subsection{Extension of the Probabilistic Risk Sensitivity: Binary Gambles with Equal or Unequal Probability Distribution}

In the previous analysis, the reference point for evaluating the gamble was the stationary status quo, denoted by $x$. A natural extension of this framework is to consider the comparison between two binary gambles that share the same probability distribution but differ in their outcome values. This generalization allows us to assess relative risk preferences across different contexts, rather than with respect to a single fixed reference point.

In particular, consider the following two binary gambles:
\begin{equation}
    A = \{a, \rho_{\text{loss}}; \, b, \rho_{\text{gain}}\}, \quad \text{with } a < b
\end{equation}
\begin{equation}
    B = \{c, \rho_{\text{loss}}; \, d, \rho_{\text{gain}}\}, \quad \text{with } c < d
\end{equation}
where $\rho_{\text{gain}} + \rho_{\text{loss}} = 1$. Let $\mathbb{E}[A]$ and $\mathbb{E}[B]$ denote the objective expected values of gambles $A$ and $B$, respectively, and let $\mathbb{E}[u(A)]$, $\mathbb{E}[u(B)]$ denote their corresponding expected utilities.

We examine how the relationship between the outcomes and the gain/loss probabilities effects the comparison between the gambles.
\begin{itemize}
    \item If $\mathbb{E}[A] > \mathbb{E}[B]$, then:
    \begin{itemize}
        \item If $b > d$, the inequality
        \[
            \frac{c - a}{b - d} < \frac{\rho_{\text{gain}}}{\rho_{\text{loss}}}
        \]
        must hold. Furthermore, if $c \leq a$, this inequality is satisfied for all $\rho_{\text{gain}} \in (0,1)$.

        \item If $b < d$, the inequality reverses:
        \[
            \frac{c - a}{b - d} > \frac{\rho_{\text{gain}}}{\rho_{\text{loss}}}
        \]
        In this case, if $c \geq a$, there exists no $\rho_{\text{gain}} \in (0,1)$ for which the inequality holds.
    \end{itemize}

    \item If $\mathbb{E}[A] < \mathbb{E}[B]$, then:
    \begin{itemize}
        \item If $b > d$, we require:
        \[
            \frac{c - a}{b - d} > \frac{\rho_{\text{gain}}}{\rho_{\text{loss}}}
        \]
        and if $c \leq a$, the inequality cannot hold for any $\rho_{\text{gain}} \in (0,1)$.

        \item If $b < d$, the inequality becomes:
        \[
            \frac{c - a}{b - d} < \frac{\rho_{\text{gain}}}{\rho_{\text{loss}}}
        \]
        and if $c \geq a$, the inequality holds for all $\rho_{\text{gain}} \in (0,1)$.
    \end{itemize}
\end{itemize}

The non-trivial cases arise when the two gambles are nested, specifically when the outcomes of one gamble lie strictly within the range of the other. These cases include: $a < c < d < b$ and $c < a < b < d$. 
These two configurations are symmetric in structure and yield analogous results. Thus, the conclusions can be unified as follows:
\begin{itemize}
    \item If
    \[
        \frac{(\text{loss}_{\text{inner}} - \text{loss}_{\text{outer}})}{(\text{gain}_{\text{outer}} - \text{gain}_{\text{inner}})} > \frac{\rho_{\text{gain}}}{\rho_{\text{loss}}}
        \quad \Rightarrow \quad \mathbb{E}[\text{inner}] > \mathbb{E}[\text{outer}]
    \]
    
    \item If
    \[
        \frac{(\text{loss}_{\text{inner}} - \text{loss}_{\text{outer}})}{(\text{gain}_{\text{outer}} - \text{gain}_{\text{inner}})} < \frac{\rho_{\text{gain}}}{\rho_{\text{loss}}}
        \quad \Rightarrow \quad \mathbb{E}[\text{inner}] < \mathbb{E}[\text{outer}]
    \]
\end{itemize}
where the following notation is used:
\begin{itemize}
    \item \textbf{inner}: the gamble whose two outcomes are more closely spaced. For instance, in the case $a < c < d < b$, gamble $B$ is the inner gamble.
    
    \item \textbf{outer}: the gamble whose outcomes span a wider range. In the case $a < c < d < b$, gamble $A$ is the outer gamble.
    
    \item \(\text{loss}_{\text{inner}}\): the loss outcome of the inner gamble. For example, $c$ in $a < c < d < b$.
    
    \item \(\text{gain}_{\text{inner}}\): the gain outcome of the inner gamble. For example, $d$ in $a < c < d < b$.
    
    \item \(\text{loss}_{\text{outer}}\): the loss outcome of the outer gamble. For example, $a$ in $a < c < d < b$.
    
    \item \(\text{gain}_{\text{outer}}\): the gain outcome of the outer gamble. For example, $b$ in $a < c < d < b$.
\end{itemize}
This formulation captures how the structure of nested gambles interacts with the probability ratio $\frac{\rho_{\text{gain}}}{\rho_{\text{loss}}}$ to determine which gamble has the higher expected value.

With a similar line of reasoning, we can extend the analysis to the subjective case, where outcomes are evaluated through a utility function $u(\cdot)$. The expected utilities of the gambles now guide the preference comparison. The following results hold:
\begin{itemize}
    \item If
    \[
        \frac{u(\text{loss}_{\text{inner}}) - u(\text{loss}_{\text{outer}})}{u(\text{gain}_{\text{outer}}) - u(\text{gain}_{\text{inner}})} > \frac{\rho_{\text{gain}}}{\rho_{\text{loss}}}
        \quad \Rightarrow \quad \mathbb{E}[u(\text{inner})] > \mathbb{E}[u(\text{outer})]
    \]
    \item If
    \[
        \frac{u(\text{loss}_{\text{inner}}) - u(\text{loss}_{\text{outer}})}{u(\text{gain}_{\text{outer}}) - u(\text{gain}_{\text{inner}})} < \frac{\rho_{\text{gain}}}{\rho_{\text{loss}}}
        \quad \Rightarrow \quad \mathbb{E}[u(\text{inner})] < \mathbb{E}[u(\text{outer})]
    \]
\end{itemize}
These inequalities describe how the subjective ranking of the gambles depends not only on the structure of the outcomes but also on the curvature of the utility function. The analysis highlights the role of risk preferences in shaping choice when objective equality in probabilities is preserved.

It is important to emphasize that since the utility function $u$ is strictly increasing, the sign relationships between outcomes remain unchanged when transitioning from the objective to the subjective evaluation. In other words, the ordering of outcomes is preserved under the transformation by $u$. Mathematically, this implies:
\begin{equation}
    \text{loss}_{\text{inner}} > [<] \text{loss}_{\text{outer}} \;\; \Leftrightarrow \;\; u(\text{loss}_{\text{inner}}) > [<] u(\text{loss}_{\text{outer}})
\end{equation}
\begin{equation}
    \text{gain}_{\text{inner}} > [<] \text{gain}_{\text{outer}} \;\; \Leftrightarrow \;\; u(\text{gain}_{\text{inner}}) > [<] u(\text{gain}_{\text{outer}})
\end{equation}
Therefore, if the two gambles are not nested and one is preferred over the other in the \textit{objective domain} (i.e., based on expected monetary values), the same preference will hold in the \textit{subjective domain} (i.e., based on expected utilities). However, the magnitude of the preference, that is, how strongly one gamble is favored, may differ due to the shape of the utility function. In particular, if the corresponding ratio is negative, the decision becomes trivial, since expected utility will always favor the gamble with strictly better outcomes.
As a result, we define a general expression that captures risk sensitivity between binary gambles around a reference (stationary) state with equal probability distributions as:
\begin{equation}
    \mathcal{RS}_{\text{gamble}}^{\text{equal}}(x, \delta_1, \delta_2, \delta_3, \delta_4)  := 
    \frac{u(x - \delta_3) - u(x - \delta_1)}{u(x + \delta_4) - u(x + \delta_2)},
    \quad \text{with } 0 \leq \delta_3 \leq \delta_1,\; 0 \leq \delta_2 \leq \delta_4
\end{equation}
This ratio reflects how differences in losses and gains are valued under the utility function. A higher value of $\mathcal{RS}_{\text{gamble}}^{\text{equal}}$ indicates a stronger aversion to gambling with a wider spread in outcomes, assuming symmetric probabilities.

Regarding the risk aversion definitions for nested gambles with equal probability distribution, we propose the following definitions:
\begin{itemize}
    \item \emph{Probabilistic symmetric nested bet aversion with equal probability distribution}: For a decision problem around a reference (stationary) state $x$ between two binary gambles $\{x-\delta_1,1-p;x+\delta_1,p\},\{x-\delta_2,1-p;x+\delta_2,p\}$ for a fixed direction $(\delta_1,\delta_2,\delta_2,\delta_1)\in\mathbb{R}_{++++}^4$, the subjective probability threshold described by a utility function $u$ is greater than the objective probability threshold described by a linear utility function
    \begin{equation}
        \mathcal{RS}_{\text{gamble}}^{\text{equal}}(x, \delta_1, \delta_2, \delta_2, \delta_1) > \frac{(x - \delta_2) - (x - \delta_1)}{(x + \delta_1) - (x + \delta_2)}
    \end{equation}
    This definition is an equivalent gamble-based approach of \emph{increasing symmetric bet aversion}.
    \item \emph{Probabilistic increasing symmetric nested bet aversion with equal probability distribution}: For a decision problem around a reference (stationary) state $x$ between two binary gambles  $\{x-\delta_1,1-p;x+\delta_1,p\},\{x-\delta_2,1-p;x+\delta_2,p\}$ for a fixed direction $(\delta_1,\delta_2,\delta_2,\delta_1)\in\mathbb{R}_{++++}^4$, the subjective probability threshold described by a utility function $u$ is greater than the objective probability threshold described by a linear utility function and also the subjective probability threshold is an increasing function of $\delta_1$
    \begin{equation}
    \begin{split}
        & \mathcal{RS}_{\text{gamble}}^{\text{equal}}(x, \delta_1, \delta_2, \delta_2, \delta_1) > \frac{(x - \delta_2) - (x - \delta_1)}{(x + \delta_1) - (x + \delta_2)} \text{ and } \\
        & \quad \quad \quad \quad \quad \frac{\partial \mathcal{RS}_{\text{gamble}}^{\text{equal}}(x, \delta_1, \delta_2, \delta_2, \delta_1)}{\partial \delta_1} > 0
    \end{split}
    \end{equation}
    The term "increasing" refers to scaling the outer gamble while holding the inner gamble fixed. This convention is consistent with the earlier stationary-state formulation, where increasing aversion is defined by scaling the risky alternative.
    \item \emph{Probabilistic non-symmetric nested bet aversion with equal probability distribution}: For a decision problem around a reference (stationary) state $x$ between two binary gambles $\{x-\delta_1,1-p;x+\delta_4,p\},\{x-\delta_3,1-p;x+\delta_2,p\},0\leq \delta_3\leq \delta_1,\;0\leq \delta_2\leq \delta_4$, the subjective probability threshold described by a utility function $u$ is greater than the objective probability threshold described by a linear utility function
    \begin{equation}
        \mathcal{RS}_{\text{gamble}}^{\text{equal}}(x, \delta_1, \delta_2, \delta_3, \delta_4) > \frac{(x - \delta_3) - (x - \delta_1)}{(x + \delta_4) - (x + \delta_2)}
    \end{equation}
    This definition is an equivalent gamble-based approach of \emph{Neilson's strong aversion}.
\end{itemize}

This approach can be further generalized to handle comparisons between two binary gambles around a reference (stationary) state with unequal probability distributions.
Let us define the following metric:
\begin{subequations}
\begin{equation}
     \mathcal{RS}_{\text{gamble}}^{\text{unequal}}(x, \delta_1, \delta_2, \delta_3, \delta_4, p_A, p_B) := 
    \frac{u(x - \delta_3) \cdot \frac{1 - p_B}{1 - p_A} - u(x - \delta_1)}{u(x + \delta_4) - u(x + \delta_2) \cdot \frac{p_B}{p_A}}
\end{equation}
\begin{equation}
    \begin{cases}
        \mathcal{RS}_{\text{gamble}}^{\text{unequal}}(x, \ldots)  > \frac{p_A}{1 - p_A} \Rightarrow \mathbb{E}[u(\text{gamble B})] > \mathbb{E}[u(\text{gamble A})] \\
         \mathcal{RS}_{\text{gamble}}^{\text{unequal}}(x, \ldots)  < \frac{p_A}{1 - p_A} \Rightarrow \mathbb{E}[u(\text{gamble B})] < \mathbb{E}[u(\text{gamble A})]
    \end{cases}
\end{equation}
\end{subequations}
where $\text{gamble A} = \{x - \delta_1,\; 1 - p_A;\; x + \delta_4,\; p_A\}$ and $\text{gamble B} = \{x - \delta_3,\; 1 - p_B;\; x + \delta_2,\; p_B\}$.
We assume the following conditions hold:
\[
u(x - \delta_1) \cdot (1 - p_A) < u(x - \delta_3) \cdot (1 - p_B), \quad 
u(x + \delta_2) \cdot p_B < u(x + \delta_4) \cdot p_A
\]
This generalized risk sensitivity ratio evaluates how much more (or less) favorable gamble B is compared to gamble A when both outcome values and probabilities differ.

Unlike the case of equal-probability gambles, where monotonicity of $u$ guarantees that objective and subjective preferences align, the decision may differ when probabilities are unequal. This divergence arises from the non-linear weighting of gains and losses in the utility space, which can shift the preference ordering between gambles.

\subsection{Dependence of Stationary State at Risk Sensitivity Metric}

In this section, we are going to investigate the influence of the stationary state at the probability threshold of the probabilistic risk sensitivity metrics, which was proposed in the previous subsections. The motivation behind this search is the Arrow-Pratt analysis, \cite{Arrow, Pratt}, about the risk aversion metrics and how the stationary state affects them. 

In our proposed framework, the dependence of the stationary state, around which a decision problem is deployed, at the risk sensitivity metric is more complicated. To start with, we provide the following lemma about the \emph{constant probabilistic risk sensitive} (CPRS) utility function at a domain. 
\begin{lemma}
    Any affine transformation of a utility function $u(x)$, $\mu + \kappa \cdot u(x)$, which utility function has the following property at a domain $\mathcal{S}$: $u(x+y) = f_1(x) \cdot f_2(y) \; \forall x,y : (x+y) \in \mathcal{S}$ (this property does not include constant functions or real inverse multipliers of $f_1,f_2$), causes the probabilistic risk sensitivity metrics to be independent from the stationary state $x \in \mathcal{S}$: $\mathcal{RS}_{stat}(x,\delta_1,\delta_2) = \mathcal{RS}_{stat}(\delta_1,\delta_2)$, $\mathcal{RS}^{equal}_{gamble}(x,\delta_1,\delta_2,\delta_3,\delta_4) = \mathcal{RS}^{equal}_{gamble}(\delta_1,\delta_2,\delta_3,\delta_4)$, $\mathcal{RS}^{unequal}_{gamble}(x,\delta_1,\delta_2,\delta_3,\delta_4) = \mathcal{RS}^{unequal}_{gamble}(\delta_1,\delta_2,\delta_3,\delta_4)$
\end{lemma}
Hence, any affine transformation of an exponential utility function satisfy the crucial property, which causes \emph{constant probabilistic risk sensitive}. For reasons of completeness and better understanding, we provide the graphical representation of probability threshold $\mathbb{P}_{gain}$ of choosing the binary gamble instead of the stationary state for each branch of K\"{o}bberling-Wakker utility function, which are exponential based. In the gain domain, the probability is given by: 
\[
    \mathbb{P}_{gain}(x,\delta) = \frac{\exp{(\alpha \cdot \delta)} - 1}{2 \cdot \sinh{(\alpha \cdot \delta)}}
\]
where $\alpha > 0$ is the curvature parameter for gains. In the loss domain, the corresponding expression is:
\[
    \mathbb{P}_{gain}(x,\delta) = \frac{1 - \exp{(- \beta \cdot \delta)}}{2 \cdot \sinh{(\beta \cdot \delta)}}
\]
where $\beta > 0$ is the curvature parameter for losses.

\begin{figure}[H]
\centering
\begin{subfigure}{.5\textwidth}
    \centering
    \includegraphics[width=\textwidth]{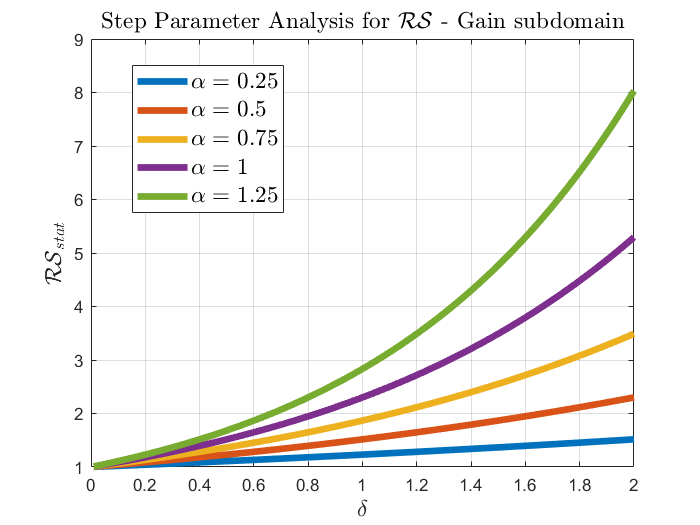}
    \caption{Gain subdomain}
    \label{fig:probability_cpt_P_gain_gain_subdomain}
\end{subfigure}%
\begin{subfigure}{.5\textwidth}
    \centering
    \includegraphics[width=\textwidth]{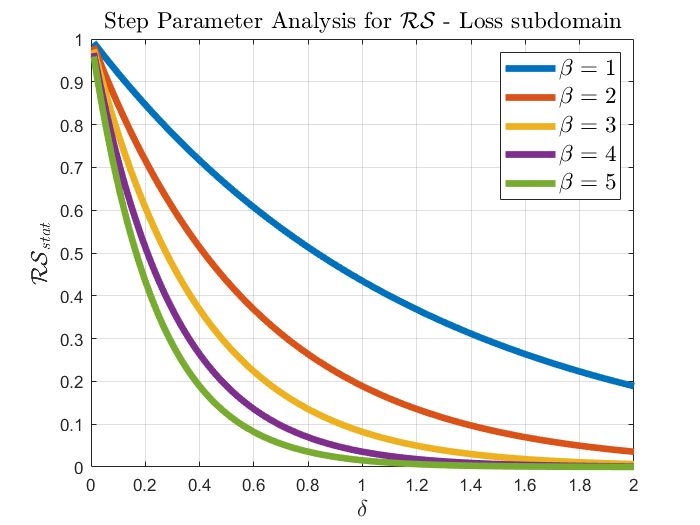}
    \caption{Loss subdomain}
    \label{fig:probability_cpt_P_gain_loss_subdomain}
\end{subfigure}
\caption[short]{Probability thresholds of choosing the binary gamble instead of the stationary state for each branch of K\"{o}bberling and Wakker utility function}
\end{figure}

\subsection{Risk Sensitivity Analysis of the EUT (Bernoulli) Utility Function Using Binary-Gamble Framework}

In this section, we apply the risk sensitivity framework of binary gambles to the classic Bernoulli utility function of EUT. This step may appear straightforward, but it provides a more detailed view of the risk behavior implied by this utility function. Moreover, this analysis yields a clearer interpretation of the risk-sensitivity metric for binary gambles with unequal probability distributions. Indeed, although the mathematical form of $\mathcal{RS}_{\text{gamble}}^{\text{unequal}}$ can be viewed as a first-order generalization of $\mathcal{RS}_{\text{gamble}}^{\text{equal}}$, it does not possess the same interpretative clarity. This limitation arises because the probabilities $p_A$ and $p_B$ are not disentangled in the expression. To obtain a more insightful interpretation, we examine how a concave utility function, as in the EUT framework, effects the probability threshold at which preference switches between two binary gambles with unequal probability distributions. Analogous results hold for the case of a convex utility function.

To start with, the EUT utility function is concave and therefore satisfies Jensen’s inequality: for any fair $N$-outcome gamble $X$ with $\mathbb{E}_{\mathcal P}[X]=x_1$, we have $\mathbb{E}_{\mathcal P}[u(X)] \le u(x_1)=u(\mathbb{E}_{\mathcal P}[X])$ (with strict inequality under strict concavity). This implies \emph{symmetric bet aversion} and \emph{Neilson's weak aversion} for symmetric and non-symmetric binary gambles respectively around $x_1$. The \emph{increasing symmetric bet aversion} and the \emph{Neilson's strong aversion} are satisfied by the concavity of the utility function. We have to underline that the \emph{symmetric bet aversion}, \emph{increasing symmetric bet aversion} and \emph{Neilson's weak aversion} are referred to the quantitative intensity approach based on the utility premium.

Regarding the proposed metrics for risk sensitivity in this work, the EUT utility function satisfies the following lemmas (for a concave $u$ and for parameters chosen so that the denominators below are positive, e.g., $\delta_3 < \delta_1$ and $\delta_2 < \delta_4$):
\begin{itemize}
    \item 
    \begin{lemma} \label{EUT_symmetric}
        $\mathcal{RS}_{stat}(x,\delta) > 1$ and $\mathcal{RS}_{stat}(x,\delta)$ is (weakly) increasing in $\delta$ (strictly increasing under strict concavity).
    \end{lemma}
    \begin{proof}
        See Appendix A.
    \end{proof}
    \item 
    \begin{lemma} \label{EUT_non_symmetric}
        $\mathcal{RS}_{stat}(x,\delta_1,\delta_2,w) > \frac{\delta_1}{\delta_2}$ and $\mathcal{RS}_{stat}(x,\delta_1,\delta_2,w)$ is (weakly) increasing in $w$ (strictly increasing under strict concavity).
    \end{lemma}
    \begin{proof}
        See Appendix A.
    \end{proof}
    \item 
    \begin{lemma} \label{EUT_Neilson}
        $\mathcal{RS}_{\text{gamble}}^{\text{equal}}(x, \delta_1, \delta_2, \delta_3, \delta_4) > \frac{(x - \delta_3) - (x - \delta_1)}{(x + \delta_4) - (x + \delta_2)}$ and $\frac{\partial \mathcal{RS}_{\text{gamble}}^{\text{equal}}(x, \delta_1, \delta_2, \delta_2, \delta_1,w)}{\partial \delta_1}>0$
    \end{lemma}
    \begin{proof}
        See Appendix A.
    \end{proof}
\end{itemize}

Next, we investigate binary gambles with unequal probability distributions. Consider two gambles
$A:\{a,1-p_A;\, b,p_A\}$ and $B:\{c,1-p_B;\, d,p_B\}$. Without loss of generality, we take gamble $A$ as the benchmark. In the absence of utility curvature (i.e., under expected-value comparison), we have
\[
\mathbb{E}[A] \gtrless \mathbb{E}[B]
\ \Leftrightarrow\
\mathbb{E}[A] \gtrless c + p_B(d-c)
\ \Leftrightarrow\
\frac{\mathbb{E}[A]-c}{d-c} \gtrless p_B,
\]
where $\mathbb{E}[A]=a(1-p_A)+bp_A$ is the expected value of gamble $A$. To ensure that the threshold has a probabilistic meaning (i.e., lies in $(0,1)$), we require $d>c$ and $c<\mathbb{E}[A]<d$ (if $d<c$, the inequality reverses when dividing by $d-c$).

Under the utility function, the corresponding comparison is
\[
\mathbb{E}[u(A)] \gtrless \mathbb{E}[u(B)]
\ \Leftrightarrow\
\mathbb{E}[u(A)] \gtrless u(c) + p_B\bigl(u(d)-u(c)\bigr)
\ \Leftrightarrow\
\frac{\mathbb{E}[u(A)]-u(c)}{u(d)-u(c)} \gtrless p_B,
\]
where $\mathbb{E}[u(A)]=u(a)(1-p_A)+u(b)p_A$ is the expected utility of gamble $A$. For the threshold to lie in $(0,1)$ we require $u(d)>u(c)$ and $u(c)<\mathbb{E}[u(A)]<u(d)$ (again, if $u(d)<u(c)$ the inequality reverses).

We investigate the relation between the two thresholds, since their comparison reveals the risk sensitivity induced by the utility function. For the geometric interpretation, let $(\epsilon_1)$ denote the chord connecting $(a,u(a))$ and $(b,u(b))$ (the locus of expected utilities of a binary gamble with outcomes $a$ and $b$ as its probability varies), and let $(\epsilon_2)$ denote the chord connecting $(c,u(c))$ and $(d,u(d))$.

There are four major cases:
\begin{itemize}
\item \textbf{Case $c < a < b < d$}: In this case, for the same value of the domain $x$, the co-domain values of the straight lines $(\epsilon_1)$ and $(\epsilon_2)$ satisfy $y_{(\epsilon_2)}(x) < y_{(\epsilon_1)}(x)$ for all $x \in [a,b]$ (with strict inequality for $x \in (a,b)$ under strict concavity). Consequently, by comparing the slopes of the segments from $c$ to $\mathbb{E}[A]$ and from $c$ to $d$, namely $\frac{\mathbb{E}[u(A)] - u(c)}{\mathbb{E}[A] - c}$ and $\frac{u(d) - u(c)}{d - c}$, we obtain
\[
    \frac{\mathbb{E}[u(A)] - u(c)}{u(d) - u(c)} > \frac{\mathbb{E}[A] - c}{d - c}
\]
for every probability distribution of gamble $A$.
    
\item \textbf{Case $a < c < d < b$}: In this case, for the same value of the domain $x$, the co-domain values for the straight lines $(\epsilon_1)$ and $(\epsilon_2)$ have the relation $y_{(\epsilon_1)}(x) < y_{(\epsilon_2)}(x)$ on the interval $[c,d]$. First, for those probability distributions of gamble $A$ such that $a < \mathbb{E}[A] < c$, we have $u(a) < \mathbb{E}[u(A)] < u(c)$, and hence gamble $B$ is preferred. Second, there exists a region with $c < \mathbb{E}[A] < d$ but $\mathbb{E}[u(A)] < u(c)$. This means that, despite the fact that for some probability distributions gamble $A$ is preferred in the risk-insensitive (expected-value) comparison, the decision maker continues to prefer gamble $B$ once risk sensitivity (via the concave EUT Bernoulli utility) is taken into account. We refer to this threshold mismatch as a risk-sensitive switching region: over this region, the preference ranking induced by a risk-sensitive utility function differs from the ranking induced by the risk-insensitive expected-value criterion. More precisely, the risk-sensitive utility function exhibits a form of inertia, in the sense that it tends to preserve the probabilistic behavior associated with the previous region, whereas the risk-insensitive criterion has already shifted to a different preference regime. Furthermore, when $c < \mathbb{E}[A] < d$ and $u(c) < \mathbb{E}[u(A)] < u(d)$, comparing the corresponding thresholds yields
\[
    \frac{\mathbb{E}[u(A)] - u(c)}{u(d) - u(c)} < \frac{\mathbb{E}[A] - c}{d - c},
\]
which indicates a tendency to prefer gamble $B$ in the risk-sensitive case relative to the risk-insensitive one. Moreover, when $d < \mathbb{E}[A]$, there exist probability distributions for which $\mathbb{E}[u(A)] < u(d)$, so that risk-sensitive switching region arises again: gamble $B$ remains preferred for some probability distributions even though gamble $A$ is strictly preferred in the risk-insensitive case.
    
\item \textbf{Case $c < a < d < b$:} In this case, for the same value of the domain $x$, the co-domain values for the straight lines $(\epsilon_1)$ and $(\epsilon_2)$ satisfy \textbf{$y_{(\epsilon_2)}(x) < y_{(\epsilon_1)}(x)$ for $\mathbb{E}[A] \in [a,\, \mathbb{E}_{\rho_0}[A))$} and \textbf{$y_{(\epsilon_1)}(x) < y_{(\epsilon_2)}(x)$ for $\mathbb{E}[A] \in (\mathbb{E}_{\rho_0}[A],\, b]$}. In the first region, the analysis is analogous to the case $c < a < b < d$, which means that gamble $A$ is preferred more strongly in the risk-sensitive case than in the risk-insensitive case. In the second region, the analysis is analogous to the case $a < c < d < b$, thus presenting a risk-sensitive switching region. Indeed, due to the linear behavior of the mean value, the probability threshold for preferring gamble $B$ increases more rapidly in the risk-insensitive case than in the risk-sensitive case, thereby creating a region where gamble $B$ remains preferred under risk sensitivity even though gamble $A$ is preferred under risk neutrality.

\item \textbf{Case $a < c < b < d$:} In this case, for the same value of the domain $x$, the co-domain values for the straight lines $(\epsilon_1)$ and $(\epsilon_2)$ satisfy $y_{(\epsilon_1)}(x) < y_{(\epsilon_2)}(x)$ for $\mathbb{E}[A] \in [a,\, \mathbb{E}_{\rho_0}[A))$ and \textbf{$y_{(\epsilon_2)}(x) < y_{(\epsilon_1)}(x)$ for $\mathbb{E}[A] \in (\mathbb{E}_{\rho_0}[A],\, b]$}. In the first region, the analysis is analogous to the case $a < c < d < b$, thus exhibiting a risk-sensitive switching region. Indeed, due to the linear behavior of the mean value, gamble $B$ appears more attractive than gamble $A$ over this interval. In the second region, the rate at which the probability threshold increases in the risk-sensitive case exceeds that of the risk-insensitive case, reflecting a stronger preference for gamble $A$ under the concave utility function.

\end{itemize}

To sum up, in the case of nested gambles, there is a tendency to prefer the nested inner gamble over the outer one, which is a property inherited from \emph{Neilson’s strong aversion}. In the cases of non-nested gambles, there is a tendency to prefer the right-shifted gamble over the left-shifted one for low values of $\mathbb{E}[A]$, with this tendency reversing for high values of $\mathbb{E}[A]$. This behavior can be interpreted as follows: the decision maker is more willing to accept the smaller gain, which appears more attainable, while effectively disregarding the more negatively impactful loss. This observation along with the risk-sensitive switching regions highlight inherent limitations in using a single functional form to fully capture risk behavior across different regions.

\subsection{Curvature of the Utility Function through the Binary Gambles Framework}

The following propositions provide an intuitive interpretation of risk sensitivity by relating it to the curvature of the utility function. Specifically, the decision maker's subjective probability threshold for accepting a binary gamble reveals whether the utility function is concave, convex, or linear in the neighborhood of the current state.

\begin{theorem} \label{concave_equivalent}
    Let $I\subseteq\mathbb{R}$ be an interval and let $u:I\to\mathbb{R}$ be strictly increasing. Then $u$ is strictly concave on $I$ if and only if, for every $x\in I$ and every $\delta_1,\delta_2>0$ such that $[x-\delta_1,x+\delta_2]\subset I$,
    \[
        \mathbb{P}_{\text{gain}}(x,\delta_1,\delta_2)>\frac{\delta_1}{\delta_1+\delta_2}.
    \]
\end{theorem}
\begin{proof}
    See Appendix A.
\end{proof}
\begin{theorem} \label{convex_equivalent}
    Let $I\subseteq\mathbb{R}$ be an interval and let $u:I\to\mathbb{R}$ be strictly increasing. Then $u$ is strictly convex on $I$ if and only if, for every $x\in I$ and every $\delta_1,\delta_2>0$ such that $[x-\delta_1,x+\delta_2]\subset I$,
    \[
        \mathbb{P}_{\text{gain}}(x,\delta_1,\delta_2)<\frac{\delta_1}{\delta_1+\delta_2}.
    \]
\end{theorem}
\begin{proof}
    See Appendix A.
\end{proof}
\begin{theorem} \label{linear_equivalent}
    Let $I\subseteq\mathbb{R}$ be an interval and let $u:I\to\mathbb{R}$ be strictly increasing. Then $u$ is affine on $I$ if and only if, for every $x\in I$ and every $\delta_1,\delta_2>0$ such that $[x-\delta_1,x+\delta_2]\subset I$,
    \[
        \mathbb{P}_{\text{gain}}(x,\delta_1,\delta_2)=\frac{\delta_1}{\delta_1+\delta_2}.
    \]
\end{theorem}
\begin{proof}
    See Appendix A.
\end{proof}

These results imply that a \textit{risk-averse} individual is characterized by a gain acceptance threshold $\mathbb{P}_{\text{gain}}$ that exceeds the risk-neutral benchmark $\frac{\delta_1}{\delta_1 + \delta_2}$, while a \textit{risk-seeking} individual has a threshold below this benchmark. A \textit{risk-neutral} decision-maker will accept the gamble precisely when $\mathbb{P}_{\text{gain}} = \frac{\delta_1}{\delta_1 + \delta_2}$, which corresponds to a linear utility function over the relevant domain.

\section{Comparison of Risk Sensitivity Metrics: Risk Premium, Probability Premium, Utility Premium and Probabilistic Risk Sensitivity}\label{Section4}

This section compares four risk-sensitivity metrics: the risk-premium framework and Arrow--Pratt ARA/RRA, the probability premium, the utility premium, and probabilistic risk sensitivity. The comparison clarifies what each metric measures and when each is appropriate for describing a decision strategy under a mean-value criterion. It also indicates which metrics can be used coherently when selecting parameters of a CPT utility function.

\subsection{Comparison between Utility Premium and Probabilistic Risk Sensitivity} \label{section:Probabilistic_Utility_Premium}

In this subsection, we are going to provide some quantitative analysis about divergence between the utility premium and the probabilistic approach of the risk sensitivity metric. Furthermore, the mathematical relation of the Arrow-Pratt ARA will be revealed in the probabilistic framework. We are going to prove that the two metrics are equivalent in the simple definition of \emph{symmetric} and \emph{non-symmetric bet aversion} but the probabilistic metric is more stringent than the utility premium metric in the \emph{increasing} definition of \emph{symmetric} and \emph{non-symmetric bet aversion}.

\subsubsection{Symmetric Binary Gambles around a Stationary State}

\begin{lemma}
    For any $x$ and $\delta>0$,
    \[
        \mathcal{RS}_{stat}(x,\delta) \gtrless 1 
        \quad \Leftrightarrow \quad 
        \mathbb{E}(x,\delta) \lessgtr u(x) \quad \Leftrightarrow \quad  \mathcal{U}(x,\delta) \lessgtr 0
    \]
    which means that the classification of the risk-sensitivity profile is equivalent whether one uses the probabilistic threshold or the utility premium for symmetric gambles around the stationary state $x$.
\end{lemma}
Next, we are going to study the \emph{increasing symmetric bet aversion}. Regarding the EUT utility function, we can observe from the proof of lemma \ref{EUT_symmetric}, that the probabilistic definition leads to the utility premium.
Focusing now on the CPT utility function and especially at the reference point, we take the limit of the first derivative of $\mathcal{RS}_{stat}$ at the reference point $x_0$ as $\delta$ tends to $0$ from the right:
\[
\begin{split}
    & \lim_{\delta \to 0^+} \mathcal{RS}'_{stat}(x_0,\delta) = \\
    &\lim_{\delta \to 0^+} \frac{ u'(x_0 - \delta)\bigl(u(x_0 + \delta) - u(x_0)\bigr) - u'(x_0 + \delta)\bigl(u(x_0) - u(x_0 - \delta)\bigr)}{\bigl( u(x_0 + \delta) - u(x_0) \bigr)^2},
\end{split}
\]
which is an indeterminate form of type $\frac{0}{0}$. By applying de l'Hôpital's rule, we obtain
\[
\begin{split}
    & \lim_{\delta \to 0^+} \mathcal{RS}'_{stat}(x_0,\delta) = \\
    & \lim_{\delta \to 0^+} \frac{1}{2\,u'(x_0 + \delta)} \left(-\,u''(x_0 - \delta) - u''(x_0 + \delta)\, \frac{u(x_0) - u(x_0 - \delta)}{u(x_0 + \delta) - u(x_0)} \right).
\end{split}
\]
The ratio inside the parentheses is again of the indeterminate form $\frac{0}{0}$, and a second application of de l'Hôpital's rule yields
\[
    \lim_{\delta \to 0^+} \mathcal{RS}'_{stat}(x_0,\delta) = -\,\frac{1}{2\,u'(x_0^+)} \left( u''(x_0^-) + u''(x_0^+)\,\frac{u'(x_0^-)}{u'(x_0^+)} \right).
\]
Rewriting this expression in terms of the Arrow–Pratt coefficient of absolute risk aversion
$A(x) = -\dfrac{u''(x)}{u'(x)}$, we obtain
\begin{equation}
    \lim_{\delta \to 0^+} \mathcal{RS}'_{stat}(x_0,\delta) = \frac{u'(x_0^-)}{2\,u'(x_0^+)}\bigl(A(x_0^-) + A(x_0^+)\bigr)
\end{equation}
where $A(\cdot)$ denotes the Arrow–Pratt absolute risk aversion. We assume that the one-sided first and second derivatives at the reference point, $u'(x_0^\pm)$ and $u''(x_0^\pm)$, are finite. By following similar process and by taking $\delta \to 0^+$, which means that both of the outcomes are in the same subdomain, we have that
\begin{equation}
    \lim_{\delta \to 0^+} \mathcal{RS}'_{stat}(x,\delta) = A(x)
\end{equation}
\begin{lemma}
    Necessary but not sufficient condition for probabilistic \emph{increasing symmetric bet aversion} at the reference point $A(x_0^-) > - A(x_0^+)$ and at the subdomain $A(x) > 0$. 
\end{lemma}
Continuing, we delve into the relation between the utility premium and the probabilistic definitions of \emph{increasing symmetric bet aversion}.
\begin{lemma}
    Suppose that $u$ satisfies the regularity conditions stated above and that $\dfrac{\partial \mathcal{U}(x_0,\delta)}{\partial \delta}$ exists for $\delta>0$. Then the following implications hold:
    \[
        0 < \frac{\partial \mathcal{U}(x_0,\delta)}{\partial \delta} 
        \;\Rightarrow\;
        \bigl( 0<\mathcal{U}(x_0,\delta) \;\Leftrightarrow\; \mathcal{RS}_{stat}(x_0,\delta) > 1 \bigr)
        \;\Rightarrow\;
        \frac{u'(x_0^-)}{u'(x_0^+)} > 1,
    \]
    and, conversely,
    \[
        \frac{u'(x_0^-)}{u'(x_0^+)} < 1
        \;\Rightarrow\; 
        \bigl( \exists\,\delta_0 > 0:\; \frac{\partial \mathcal{U}(x_0,\delta)}{\partial \delta}\big|_{\delta = \delta_0} < 0 \bigr).
    \]
    In words, an increasing quantitative measure $\mathcal{U}(x_0,\delta)$ corresponds to rejection of symmetric gambles around $x_0$ (i.e., $\mathcal{RS}_{stat}(x_0,\delta)>1$) and requires a larger marginal utility on the loss side than on the gain side, $\frac{u'(x_0^-)}{u'(x_0^+)} > 1$. Conversely if $\frac{u'(x_0^-)}{u'(x_0^+)} < 1$, then for some nonzero stake $\delta_0$ the symmetric gamble is locally attractive (i.e. $\mathcal{RS}_{stat}(x_0,\delta_0) < 1$), which is reflected in a locally decreasing behavior of $\mathcal{U}(x_0,\delta)$ at $\delta_0$.  
\end{lemma}
\begin{proof}
    Assume that $\dfrac{u'(x_0^-)}{u'(x_0^+)} < 1$, i.e., $u'(x_0^-) < u'(x_0^+)$.  
    The derivative of the utility premium with respect to $\delta$ is
    \[
        \frac{\partial \mathcal{U}(x_0,\delta)}{\partial \delta}
        = -\frac{1}{2}\bigl(u'(x_0+\delta) - u'(x_0-\delta)\bigr).
    \]
    Taking the limit as $\delta \to 0^+$ and using the left and right derivatives at $x_0$, we obtain
    \[
        \lim_{\delta \to 0^+} \frac{\partial \mathcal{U}(x_0,\delta)}{\partial \delta}
        = -\frac{1}{2}\bigl(u'(x_0^+) - u'(x_0^-)\bigr) < 0,
    \]
    where the inequality follows from $u'(x_0^-) < u'(x_0^+)$.  
    By continuity of the one-sided derivatives in a neighborhood of $x_0$, it follows that there exists $\delta_0 > 0$ such that
    \[
        \left.\frac{\partial \mathcal{U}(x_0,\delta)}{\partial \delta}\right|_{\delta = \delta_0} < 0,
    \]
    which proves the claim.
\end{proof}
\begin{lemma}
    If $\dfrac{u'(x_0^-)}{u'(x_0^+)} > 1$ and 
    $\dfrac{\partial \mathcal{RS}_{stat}(x_0,\delta)}{\partial \delta} > 0$
    , then $ \mathcal{RS}_{stat}(x_0,\delta) > 1.$
\end{lemma}
\begin{lemma}
    If $\dfrac{u'(x_0^-)}{u'(x_0^+)} > 1$ and 
    $\dfrac{\partial \mathcal{RS}_{stat}(x_0,\delta)}{\partial \delta} > 0$
    , then $0 < \dfrac{\partial \mathcal{U}(x_0,\delta)}{\partial \delta}  \quad \text{for all such } \delta>0.$
\end{lemma}
\begin{proof}
    By taking the differentiability rule and set the nominator positive, we have that
    \[
        \frac{u'(x_0 - \delta)}{u'(x_0 + \delta)} >  \mathcal{RS}_{stat}(x_0,\delta) 
    \]
    By applying previous lemma, we have that 
    \[
        \frac{u'(x_0 - \delta)}{u'(x_0 + \delta)} > 1 \Rightarrow 0 < \frac{\partial \mathcal{U}(x_0,\delta)}{\partial \delta} 
    \]
\end{proof}
\begin{lemma}
    The condition $ 0 < \dfrac{\partial \mathcal{U}(x_0,\delta)}{\partial \delta}$
    does not, in general, imply that $0 < \dfrac{\partial \mathcal{RS}_{stat}(x_0,\delta)}{\partial \delta}.$
\end{lemma}
Finally, we provide an example for better understanding of the last lemma. Based on the Neilson-type construction for the K\"{o}bberling-Wakker utility function, following ~\cite{Neilson_2002_A,Kobberling_Wakker} and \cite{vaidanis2025Theoretical_Part_1_Neilson_extension_KW}, we assume the following S-shape utility function
\[
    u(x) = \left\{
        \begin{array}{ll}
            \dfrac{1 - \exp\left( - \alpha (x - x_0)  \right)}{\alpha} & \text{for}~ x_0 \leq x \\
            - \lambda_1 \dfrac{1 - \exp\left( - \beta (x_0 - x)  \right)}{\beta} + \lambda_2 (x - x_0) & \text{for}~ x < x_0 \\
        \end{array} 
    \right. 
\]
which satisfies the \emph{Neilson's strong aversion} at the reference point for $\lambda_2 > 1$. If we assume $\alpha = 0.5$, $\beta = 2.5$, $\lambda_1 = 1$ and $\lambda_2 = 1.5$, we have that $\lim_{\delta \to 0^+} \mathcal{RS}'_{stat}(x_0,\delta) < 0$, which means that despite the fact that the utility premium version of \emph{increasing symmetric bet aversion} holds true, the probabilistic version of \emph{increasing symmetric bet aversion} is not satisfied.

\subsubsection{Non-Symmetric Binary Gambles around a Stationary State}

\begin{lemma}
    For any $x\in\mathbb{R}$ and $\delta_1,\delta_2,w>0$,
    \[
        \mathcal{RS}_{stat}(x,\delta_1,\delta_2,w) \gtrless \frac{\delta_1}{\delta_2}
        \quad \Leftrightarrow \quad
        \mathbb{E}(x,\delta_1,\delta_2,w) \lessgtr u(x) \quad \Leftrightarrow \quad \mathcal{U}(x,\delta_1,\delta_2,w) \gtrless 0,
    \]
    which means that the classification of the risk-sensitivity profile is equivalent whether one uses the probabilistic definition of risk sensitivity or the utility premium definition for non-symmetric gambles around the stationary state $x$.
\end{lemma}
Next, we are going to study the \emph{increasing non-symmetric bet aversion}. Regarding the EUT utility function, we can observe from the proof of lemma \ref{EUT_non_symmetric} that the probabilistic definitions leads to the utility premium approach.
Focusing now on the CPT utility function and especially at the reference point, we take the limit of the first derivative of $\mathcal{RS}_{stat}$ at the reference point $x_0$ as $w$ tends to $0$ from the right:
\[
\begin{split}
    & \lim_{w \to 0^+} \mathcal{RS}'_{stat}(x_0,\delta_1,\delta_2,w) = \\
    & \lim_{w \to 0^+} \frac{ \delta_1\,u'(x_0 - w\delta_1)\bigl(u(x_0 + w\delta_2) - u(x_0)\bigr) - \delta_2\,u'(x_0 + w\delta_2)\bigl(u(x_0) - u(x_0 - w\delta_1)\bigr)}{ \bigl(u(x_0 + w\delta_2) - u(x_0)\bigr)^2}.
\end{split}
\]
This expression is an indeterminate form of type $\frac{0}{0}$. By applying de l'Hôpital's rule, we obtain
\[
\begin{split}
    & \lim_{w \to 0^+} \mathcal{RS}'_{stat}(x_0,\delta_1,\delta_2,w) = \\ 
    & \lim_{w \to 0^+} \frac{1}{2\,\delta_2\,u'(x_0 + w\delta_2)} \left(-\delta_1^2 u''(x_0 - w\delta_1) -
    \delta_2^2 u''(x_0 + w\delta_2)\, \frac{u(x_0) - u(x_0 - w\delta_1)}{u(x_0 + w\delta_2) - u(x_0)} \right).
\end{split}
\]
The ratio inside the parentheses is again of indeterminate type $\frac{0}{0}$, and a second application of de l'Hôpital's rule yields
\[
    \lim_{w \to 0^+} \mathcal{RS}'_{stat}(x_0,\delta_1,\delta_2,w) = -\frac{1}{2\,\delta_2\,u'(x_0^+)} \left( \delta_1^2 u''(x_0^-) +
    \delta_1\delta_2\,u''(x_0^+)\,\frac{u'(x_0^-)}{u'(x_0^+)}
    \right).
\]
Rewriting the latter in terms of the Arrow–Pratt coefficient of absolute risk aversion,
$A(x) = -\dfrac{u''(x)}{u'(x)}$, we obtain
\[
    \lim_{w \to 0^+} \mathcal{RS}'_{stat}(x_0,\delta_1,\delta_2,w) = \frac{u'(x_0^-)}{2\,u'(x_0^+)}\, \frac{\delta_1}{\delta_2}\, \bigl(\delta_1 A(x_0^-) + \delta_2 A(x_0^+)\bigr),
\]
where $A(\cdot)$ denotes the Arrow–Pratt absolute risk aversion. We assume that the one-sided first and second derivatives at the reference point, $u'(x_0^\pm)$ and $u''(x_0^\pm)$, are finite. By following similar process and by taking $w \to 0^+$, which means that both of the outcomes are in the same subdomain, we have that
\begin{equation}
    \lim_{w \to 0^+} \mathcal{RS}'_{stat}(x,\delta_1,\delta_2,w) = \frac{A(x)}{2}\, \frac{\delta_1}{\delta_2}\, \bigl(\delta_1 + \delta_2 \bigr),
\end{equation}
\begin{lemma}
    Necessary but not sufficient condition for probabilistic \emph{increasing non-symmetric bet aversion} at the reference point $A(x_0^-) > - A(x_0^+)\,\frac{\delta_2}{\delta_1}$ and at the subdomain $A(x) > 0$. 
\end{lemma}
As we can observe, any S-shape utility function will not be probabilistically \emph{increasing non-symmetric bet averse} even though it satisfies \emph{Neilson's strong aversion}. 
Continuing, we delve into the relation between the utility premium and the probabilistic definitions of \emph{increasing non-symmetric bet aversion}.
\begin{lemma}
    We have the following implications:
    \[
    \begin{split}
        & 0 < \frac{\partial \mathcal{U}(x_0,\delta_1,\delta_2,w)}{\partial w}
        \;\Rightarrow\;
        \bigl( 0 < \mathcal{U}(x_0,\delta_1,\delta_2,w) 
        \;\Leftrightarrow\;
        \mathcal{RS}_{stat}(x_0,\delta_1,\delta_2,w)>\frac{\delta_1}{\delta_2} \bigr)
        \\ & \;\Rightarrow\;
        \frac{u'(x_0^-)}{u'(x_0^+)} > 1,
    \end{split}
    \]
    and, conversely,
    \[
        \frac{u'(x_0^-)}{u'(x_0^+)} < 1
        \;\Rightarrow\;
        \bigl( \exists\,w_0 > 0 : \left.\frac{\partial \mathcal{U}(x_0,\delta_1,\delta_2,w)}{\partial w}\right|_{w = w_0} < 0 \bigr).
    \]
    In words, an increasing quantitative measure $\mathcal{U}(x_0,\delta_1,\delta_2,w)$ corresponds to rejection of non-symmetric gambles around $x_0$ (i.e., $\mathcal{RS}_{stat}(x_0,\delta_1,\delta_2,w)>\frac{\delta_1}{\delta_2}$) and requires a larger marginal utility on the loss side than on the gain side, $\frac{u'(x_0^-)}{u'(x_0^+)} > 1$. Conversely if $\frac{u'(x_0^-)}{u'(x_0^+)} < 1$, then for some nonzero stake $w_0$ the non-symmetric gamble is locally attractive (i.e. $\mathcal{RS}_{stat}(x_0,\delta_1,\delta_2,w_0) < \frac{\delta_1}{\delta_2}$), which is reflected in a locally decreasing behavior of $\mathcal{U}(x_0,\delta_1,\delta_2,w)$ at $w_0$. 
\end{lemma}
\begin{proof}
    Assume that $\dfrac{u'(x_0^-)}{u'(x_0^+)} < 1$, i.e., $u'(x_0^-) < u'(x_0^+)$. The derivative of the utility premium with respect to $w$ gives
    \[
        \frac{\partial \mathcal{U}(x_0,\delta_1,\delta_2,w)}{\partial w}
        = - \frac{\delta_1\delta_2}{\delta_1 + \delta_2}
        \bigl(u'(x_0 + w\delta_2) - u'(x_0 - w\delta_1)\bigr).
    \]
    Letting $w \to 0^+$ and using the one-sided derivatives at $x_0$, we obtain
    \[
        \lim_{w \to 0^+} \frac{\partial \mathcal{U}(x_0,\delta_1,\delta_2,w)}{\partial w}
        = - \frac{\delta_1\delta_2}{\delta_1 + \delta_2}\bigl(u'(x_0^+) - u'(x_0^-)\bigr) < 0,
    \]
    where the inequality follows from $u'(x_0^-) < u'(x_0^+)$. By continuity of the derivative in a neighborhood of $w=0$, there exists $w_0 > 0$ such that
    \[
        \frac{\partial \mathcal{U}(x_0,\delta_1,\delta_2,w)}{\partial w}\Big|_{w = w_0} < 0,
    \]
    which proves the claim.
\end{proof}
\begin{lemma}
    If $\dfrac{u'(x_0^-)}{u'(x_0^+)} > 1
        \quad \text{and} \quad
        \dfrac{\partial \mathcal{RS}_{stat}(x_0,\delta_1,\delta_2,w)}{\partial w} > 0
        ,$
    then $\mathcal{RS}_{stat}(x_0,\delta_1,\delta_2,w) > \frac{\delta_1}{\delta_2}$
\end{lemma}
\begin{lemma}
    If $\dfrac{u'(x_0^-)}{u'(x_0^+)} > 1
        \quad \text{and} \quad
        \dfrac{\partial \mathcal{RS}_{stat}(x_0,\delta_1,\delta_2,w)}{\partial w} > 0
        ,$
    then $\dfrac{\partial \mathcal{U}(x_0,\delta_1,\delta_2,w)}{\partial w} > 0$
    for all such $w>0.$
\end{lemma}
\begin{proof}
    By taking the differentiability rule and set the nominator positive, we have that
    \[
        \frac{u'(x_0 - w \cdot \delta_1)}{u'(x_0 + w \cdot \delta_2)} >  \mathcal{RS}_{stat}(x_0,\delta_1,\delta_2,w) \cdot \frac{\delta_2}{\delta_1}
    \]
    By applying previous lemma, we have that 
    \[
        \frac{u'(x_0 - w \cdot \delta_1)}{u'(x_0 + w \cdot \delta_2)} > 1 \Rightarrow \frac{\partial \mathcal{U}(x_0,\delta_1,\delta_2,w)}{\partial w} > 0
    \]
\end{proof}
\begin{lemma}
    The condition
    $\dfrac{\partial \mathcal{U}(x_0,\delta_1,\delta_2,w)}{\partial w} > 0$
    does not, in general, imply that \newline
    $\dfrac{\partial \mathcal{RS}_{stat}(x_0,\delta_1,\delta_2,w)}{\partial w} > 0.$
\end{lemma}
Finally, we provide an example for better understanding of the last lemma. Based on the Neilson-type construction for the K\"{o}bberling-Wakker utility function, following ~\cite{Neilson_2002_A,Kobberling_Wakker} and \cite{vaidanis2025Theoretical_Part_1_Neilson_extension_KW}, we assume the following S-shape utility function
\[
    u(x) = \left\{
        \begin{array}{ll}
            \dfrac{1 - \exp\left( - \alpha (x - x_0)  \right)}{\alpha} & \text{for}~ x_0 \leq x \\
            - \lambda_1 \dfrac{1 - \exp\left( - \beta (x_0 - x)  \right)}{\beta} + \lambda_2 (x - x_0) & \text{for}~ x < x_0 \\
        \end{array} 
    \right. 
\]
which satisfies the \emph{Neilson's strong aversion} at the reference point for $\lambda_2 > 1$. If we assume $\alpha = 0.5$, $\beta = 2.5$, $\lambda_1 = 1$, $\lambda_2 = 1.5$, $\delta_1 = 1$ and $\delta_2 = 2$, we have that $\lim_{\delta \to 0^+} \mathcal{RS}'_{stat}(x_0,\delta_1,\delta_2,w) < 0$, which means that despite the fact that the utility premium version of \emph{increasing non-symmetric bet aversion} holds true, the probabilistic version of \emph{increasing non-symmetric bet aversion} is not satisfied.

\subsection{Analysis of Probability Premium and Comparison with Probabilistic Risk Sensitivity} \label{section:Probabilistic_Probability_Premium}

In this subsection, we examine the conceptual and mathematical relation between the probabilistic risk-sensitivity metric and the probability-premium framework. On the one hand, probabilistic risk sensitivity describes the decision strategy; on the other hand, the probability premium provides a construction rule for fair decision problems under the utility transformation. We show that there is an injective mathematical relation between the two metrics.

\subsubsection{Extension of Probability Premium Framework}

To start with, for the case of a symmetric binary gamble around a stationary state, we reformulate the expression for the probability premium $\gamma$ as follows:
\begin{equation} \label{eq:probability_premium_symmetric}
\begin{split}
    \gamma(x,\delta) = & \dfrac{u(x) - \left[\frac{1}{2}u(x-\delta) + \frac{1}{2}u(x+\delta)\right]}{u(x+\delta) -  u(x-\delta)} \\
    = & \frac{u(x) - \mathbb{E}(x,\delta)}{u(x+\delta) -  u(x-\delta)}
\end{split}
\end{equation}
The first derivative of $\gamma$ with respect to $\delta$ is defined by the relation:
\begin{equation}
    \frac{\partial \gamma(x,\delta)}{\partial \delta} = \frac{- \frac{\partial \mathbb{E}(x,\delta)}{\partial \delta} \cdot [u(x+\delta) -  u(x-\delta)] - [u'(x+\delta) +  u'(x-\delta)] \cdot [u(x) - \mathbb{E}(x,\delta)]}{[u(x+\delta) -  u(x-\delta)]^2}
\end{equation}
From the expression of derivative, we can extract the following lemma:
\begin{lemma}
    If the probability premium $\gamma(x,\delta)$ is an increasing function of $\delta$, then the utility premium $\mathcal{U}(x,\delta)$ is an increasing function of $\delta$ under the condition of $\mathcal{U}(x,\delta) > 0$. The opposite direction does not hold true in general.
\end{lemma}

In addition, we extend the definition of the probability premium to non-symmetric binary gambles around a stationary state as follows:
\begin{definition}
    For a decision setup between a stationary state $x$ and a symmetric binary gamble around it $\{x-\delta,\frac{1}{2};x+\delta,\frac{1}{2}\}$, the probability premium $\gamma$ is the excess in winning probability over the fair-gamble condition that makes the agent indifferent between the stationary state and the binary gamble.
    \begin{equation}
    \begin{split}
        u(x) = &\left(\frac{\delta_1}{\delta_1 + \delta_2} + \gamma(x,\delta_1,\delta_2,w)\right)u(x+w\delta_2) \\
        &+\left(\frac{\delta_2}{\delta_1 + \delta_2} - \gamma(x,\delta_1,\delta_2,w)\right)u(x-w\delta_1)
    \end{split}
    \end{equation}
\end{definition}
Solving for the probability premium $\gamma$ gives:
\begin{equation} \label{eq:probability_premium_non_symmetric}
    \gamma(x,\delta_1,\delta_2,w) = \frac{u(x) - \mathbb{E}(x,\delta_1,\delta_2,w)}{u(x+w\delta_2) - u(x-w\delta_1)}
\end{equation}
The first derivative of $\gamma$ with respect to $w$ is defined by the relation:
\begin{subequations}
    \begin{equation}
        \frac{\partial \gamma(x,\delta_1,\delta_2,w)}{\partial w} = \frac{B(x,\delta_1,\delta_2,w)}{[u(x+w\delta_2) - u(x-w\delta_1)]^2}
    \end{equation}
    \begin{equation}
    \begin{split}
        B(x,\delta_1,\delta_2,w) = & -\frac{\partial \mathbb{E}(x,\delta_1,\delta_2,w)}{\partial w} \cdot [u(x+w\delta_2) - u(x-w\delta_1)] \\
        & - [\delta_2 u'(x+w\delta_2) + \delta_1 u'(x-w\delta_1)] \cdot [u(x) - \mathbb{E}(x,\delta_1,\delta_2,w)]
    \end{split}  
    \end{equation}
\end{subequations}
From the expression of derivative, we can extract the following lemma:
\begin{lemma}
    If the probability premium $\gamma(x,\delta_1,\delta_2,w)$ is an increasing function of $w$, then the utility premium $\mathcal{U}(x,\delta_1,\delta_2,w)$ is an increasing function of $\delta$ under the condition of $\mathcal{U}(x,\delta_1,\delta_2,w) > 0$. The opposite direction does not hold true in general.
\end{lemma}

From the above mathematical expression for symmetric and non-symmetric case, it can be observed that the mathematical expression of $\gamma$ is a normalized version of quantitative intensity framework for symmetric and non-symmetric gambles around a stationary state.

Moreover, we extend the probability-premium definition to two nested binary gambles around a reference (stationary) state with the same probability distribution.
\begin{definition}
    For a decision setup between two binary gambles with the same probability distribution around a reference (stationary) state $x$, $\{x-\delta_1,1-p;x+\delta_4,p\}$ and $\{x-\delta_3,1-p;x+\delta_2,p\}$ with $0\leq\delta_3\leq\delta_1$ and $0\leq\delta_2\leq\delta_4$, the probability premium $\gamma$ is the excess in winning probability from the fair-gamble condition that makes the agent indifferent between the stationary state and the binary gamble.
    \begin{subequations}
        \begin{equation}
            (p+\gamma) \cdot u(x+\delta_4) + (1-p-\gamma) \cdot u(x-\delta_1) = (p+\gamma) \cdot u(x+\delta_2) + (1-p-\gamma) \cdot u(x-\delta_3)
        \end{equation}
        \begin{equation}
            p = \dfrac{\delta_1 - \delta_3}{(\delta_4 - \delta_2) + (\delta_1 - \delta_3)}
        \end{equation}
    \end{subequations}
\end{definition}
Solving for the probability premium $\gamma$ gives:
\begin{subequations} \label{eq:probability_premium_equal_distribution}
    \begin{equation}
        \gamma(x,\delta_1,\delta_2,\delta_3,\delta_4) = \dfrac{(1-p) \cdot \left[ \dfrac{u(x-\delta_3) - u(x-\delta_1)}{u(x+\delta_4) - u(x+\delta_2)} - \dfrac{p}{1-p}\right]}{\dfrac{u(x-\delta_3) - u(x-\delta_1)}{u(x+\delta_4) - u(x+\delta_2)} + 1}
    \end{equation}
    \begin{equation}
        p = \dfrac{\delta_1 - \delta_3}{(\delta_4 - \delta_2) + (\delta_1 - \delta_3)}
    \end{equation}
\end{subequations} 

\subsubsection{Probability Risk Sensitivity and Probability Premium under EUT Framework}

In the literature, the probability premium describes the construction process of creating fair gambles under the presence of the utility function for a specific setup of outcomes and probability distribution. On the other hand, the probabilistic risk sensitivity metric describes the decision strategy for a mean value criterion under a given setup of outcomes. If we observe more carefully the mathematical expression for the probability premium of the previous section and reformulate them appropriately, we can extract the following interesting results:
\begin{itemize}
    \item \textit{Symmetric binary gambles around a stationary state}
        \begin{equation}
            \mathcal{RS}_{stat}(x,\delta) = \dfrac{\frac{1}{2} + \gamma}{\frac{1}{2} - \gamma} \Leftrightarrow \gamma(x,\delta) = \frac{1}{2} \cdot \dfrac{\mathcal{RS}_{stat} - 1}{\mathcal{RS}_{stat} + 1}
        \end{equation}
    \item \textit{Non-symmetric binary gambles around a stationary state}
        \begin{equation}
            \mathcal{RS}_{stat}(x,\delta_1,\delta_2,w) = \dfrac{\frac{\delta_1}{\delta_1 + \delta_2} + \gamma}{\frac{\delta_2}{\delta_1 + \delta_2} - \gamma} \Leftrightarrow \gamma(x,\delta) = \frac{\delta_2}{\delta_1 + \delta_2} \cdot \dfrac{\mathcal{RS}_{stat} - \frac{\delta_1}{\delta_2}}{\mathcal{RS}_{stat} + 1}
        \end{equation}
    \item \textit{Nested binary gambles with the same probability distribution around a reference (stationary) state}
    \begin{subequations} 
        \begin{equation}
        \begin{split}
            & \mathcal{RS}^{equal}_{gamble}(x,\delta_1,\delta_2,\delta_3,\delta_4) = \dfrac{\gamma + p}{(1-p) - \gamma} \Leftrightarrow \\ & \gamma(x,\delta_1,\delta_2,\delta_3,\delta_4) = \dfrac{(1-p) \cdot \left[ \mathcal{RS}^{equal}_{gamble} - \frac{p}{1-p} \right]}{\mathcal{RS}^{equal}_{gamble} + 1}
            \end{split}
        \end{equation}
        \begin{equation}
            p = \dfrac{\delta_1 - \delta_3}{(\delta_4 - \delta_2) + (\delta_1 - \delta_3)}
        \end{equation}
    \end{subequations}
\end{itemize}
An important observation is that, if the probability distribution of the decision problem is known under the fair-gamble condition before applying the utility function, then information about the decision strategy $\mathcal{RS}$ can be recovered from the construction parameter $\gamma$, and vice versa. In addition, $\mathcal{RS}$ and $\gamma$ have the same monotonicity with respect to $\delta$ for symmetric binary gambles around a stationary state and with respect to $w$ for non-symmetric binary gambles around a stationary state.

\subsubsection{Probability Risk Sensitivity and Probability Premium under the RDU Framework}

We now consider cumulative probability weighting under the framework of rank-dependent utility (RDU). The following analysis applies in two cases. First, it applies when all outcomes belong to a single CPT subdomain, so that CPT reduces locally to an RDU-type evaluation using either the gain-side or the loss-side probability weighting function. Second, it applies to cross-domain CPT comparisons only when the distorted gain and loss weights can be represented through a single complementary weighting rule, as in RDU.

The calculation of the probability premium follows the same basic threshold construction as in the EUT case, with the difference that the objective probability is first transformed by the probability weighting function.

\begin{itemize}
\item \textit{Symmetric binary gambles around a stationary state.}
For the binary gamble
\[
    \{x-\delta,1-p;\ x+\delta,p\},
    \qquad p=\frac{1}{2},
\]
the indifference condition under RDU is
\[
    u(x)
    =
    \omega(p+\gamma)u(x+\delta)
    +
    \bigl(1-\omega(p+\gamma)\bigr)u(x-\delta).
\]
Equivalently,
\[
    \omega(p+\gamma)
    =
    \frac{u(x)-u(x-\delta)}{u(x+\delta)-u(x-\delta)}
    \equiv
    P_{\rm gain}(x,\delta).
\]
If $\omega$ is strictly increasing, then
\[
    \gamma(x,\delta)
    =
    \omega^{-1}\!\left(P_{\rm gain}(x,\delta)\right)-p .
\]
The corresponding probabilistic risk-sensitivity ratio is
\[
    RS_{\rm stat}(x,\delta)
    =
    \frac{P_{\rm gain}(x,\delta)}{1-P_{\rm gain}(x,\delta)}
    =
    \frac{u(x)-u(x-\delta)}{u(x+\delta)-u(x)} .
\]

\item \textit{Non-symmetric binary gambles around a stationary state.}

For the binary gamble
\[
    \{x-w\delta_1,1-p;\ x+w\delta_2,p\},
    \qquad
    p=\frac{\delta_1}{\delta_1+\delta_2},
\]
the RDU indifference condition is
\[
    u(x)
    =
    \omega(p+\gamma)u(x+w\delta_2)
    +
    \bigl(1-\omega(p+\gamma)\bigr)u(x-w\delta_1).
\]
Thus,
\[
    \omega(p+\gamma)
    =
    \frac{u(x)-u(x-w\delta_1)}{u(x+w\delta_2)-u(x-w\delta_1)}
    \equiv
    P_{\rm gain}(x,\delta_1,\delta_2,w).
\]
Hence,
\[
    \gamma(x,\delta_1,\delta_2,w)
    =
    \omega^{-1}\!\left(P_{\rm gain}(x,\delta_1,\delta_2,w)\right)-p ,
\]
and
\[
    RS_{\rm stat}(x,\delta_1,\delta_2,w)
    =
    \frac{P_{\rm gain}(x,\delta_1,\delta_2,w)}{1-P_{\rm gain}(x,\delta_1,\delta_2,w)}
    =
    \frac{u(x)-u(x-w\delta_1)}{u(x+w\delta_2)-u(x)} .
\]

\item \textit{Nested binary gambles with the same probability distribution around a reference state.}

Consider the nested binary gambles
\[
    A=\{x-\delta_1,1-p;\ x+\delta_4,p\},
    \qquad
    B=\{x-\delta_3,1-p;\ x+\delta_2,p\},
\]
where
\[
    0\leq \delta_3<\delta_1,\qquad
    0\leq \delta_2<\delta_4 .
\]
The objective fair probability threshold, i.e., the probability that makes the two gambles have the same expected value without utility distortion, is
\[
    p=
    \frac{\delta_1-\delta_3}{(\delta_4-\delta_2)+(\delta_1-\delta_3)} .
\]
Under RDU, the indifference condition is
\[
\begin{aligned}
&\omega(p+\gamma)u(x+\delta_4)
+
\bigl(1-\omega(p+\gamma)\bigr)u(x-\delta_1)  \\
&\qquad =
\omega(p+\gamma)u(x+\delta_2)
+
\bigl(1-\omega(p+\gamma)\bigr)u(x-\delta_3).
\end{aligned}
\]
Solving for the distorted gain weight gives
\[
    \omega(p+\gamma)
    =
    \frac{
    u(x-\delta_3)-u(x-\delta_1)
    }{
    \bigl(u(x+\delta_4)-u(x+\delta_2)\bigr)
    +
    \bigl(u(x-\delta_3)-u(x-\delta_1)\bigr)
    },
\]
where the right part of the previous equation is equal to $P_{\rm gain}(x,\delta_1,\delta_2,\delta_3,\delta_4)$. Therefore,
\[
    \gamma(x,\delta_1,\delta_2,\delta_3,\delta_4)
    =
    \omega^{-1}\!\left(
    P_{\rm gain}(x,\delta_1,\delta_2,\delta_3,\delta_4)
    \right)-p .
\]
The corresponding probabilistic risk-sensitivity ratio is
\[
    RS_{\rm gamble}^{\rm equal}
    (x,\delta_1,\delta_2,\delta_3,\delta_4)
    =
    \frac{
    u(x-\delta_3)-u(x-\delta_1)
    }{
    u(x+\delta_4)-u(x+\delta_2)
    } .
\]

\end{itemize}

The calculation of probabilistic risk sensitivity is therefore the same as in the EUT framework at the level of the utility-induced probability threshold. The difference is that, under RDU, the probability premium is obtained by mapping this threshold through the inverse probability weighting function. Thus, for a fixed decision setup and a fixed utility function, the computation of $RS$ precedes the computation of $\gamma$.

This distinction is important. The probabilistic risk-sensitivity metric $RS$ isolates the effect of the utility function on the probability threshold that governs the decision. The probability weighting function then determines whether the distorted probability satisfies this threshold. By contrast, the probability premium $\gamma$ combines both effects: the effect of the utility function and the effect of probability distortion.
We next consider rank-dependent utility (RDU), where a single probability weighting function $\omega$ is applied to the relevant cumulative probability. The following analysis applies directly when all outcomes lie in a single CPT subdomain, so that CPT reduces to an RDU-type evaluation using either $\omega^+$ or $\omega^-$. It also applies to cross-domain comparisons only under a valid single-weighting representation, i.e., when the gain and loss decision weights are complementary.

\subsubsection{Probability Premium and Decision-Weighting Risk Sensitivity under the CPT Framework}
We now consider the genuinely CPT case in which the relevant outcomes lie on different sides of the reference point and gains and losses are weighted by two possibly different probability weighting functions, $\omega^+$ and $\omega^-$. In this case, the gain and loss decision weights are not generally complementary:
\[
\omega^+(p)+\omega^-(1-p)\neq 1 .
\]
Therefore, unlike in EUT or RDU, the indifference condition cannot, in general, be described by a one-dimensional probability threshold. Instead, it is more naturally described by a two-dimensional decision-weighting frontier.

For a non-symmetric binary gamble around a stationary state $x$,
\[
    \{x-w\delta_1,1-p;\ x+w\delta_2,p\},
\]
the geometric locus of decision weights that make the agent indifferent between the stationary state and the gamble is
\[
\mathcal S(x,\delta_1,\delta_2,w)
=
\Bigl\{
(y,z)\in[0,1]^2:
u(x)=y\,u(x+w\delta_2)+z\,u(x-w\delta_1)
\Bigr\}.
\]
Equivalently, whenever $u(x+w\delta_2)\neq 0$, this frontier can be written as
\[
    y=
    \frac{
    u(x)-z\,u(x-w\delta_1)
    }{
    u(x+w\delta_2)
    } .
\]
If the stationary state coincides with the reference point, i.e., $x=x_0$ and $u(x_0)=0$, then the frontier reduces to the decision-weight ratio
\[
    \frac{y}{z}
    =
    \frac{-u(x_0-w\delta_1)}{u(x_0+w\delta_2)}
    =
    RS_{\rm stat}(x_0,\delta_1,\delta_2,w).
\]
Thus, in the CPT framework, the usual probabilistic risk-sensitivity ratio is better interpreted as a decision-weight ratio.

For nested binary gambles with a common probability distribution, consider
\[
    A=\{x-\delta_1,1-p;\ x+\delta_4,p\},
    \qquad
    B=\{x-\delta_3,1-p;\ x+\delta_2,p\},
\]
where $0\leq\delta_3<\delta_1$ and $0\leq\delta_2<\delta_4$. The geometric locus of decision weights that make the agent indifferent between the two gambles is
\[
\mathcal S(x,\delta_1,\delta_2,\delta_3,\delta_4)
=
\Bigl\{
(y,z)\in[0,1]^2:
y\,u(x+\delta_4)+z\,u(x-\delta_1)
=
y\,u(x+\delta_2)+z\,u(x-\delta_3)
\Bigr\}.
\]
For $z>0$, this condition is equivalent to
\[
    \frac{y}{z}
    =
    \frac{
    u(x-\delta_3)-u(x-\delta_1)
    }{
    u(x+\delta_4)-u(x+\delta_2)
    }
    \equiv
    RS_{\rm gamble}^{\rm equal}(x,\delta_1,\delta_2,\delta_3,\delta_4).
\]
The geometric locus $\mathcal S$ can therefore be interpreted as a decision-weighting front threshold, and the corresponding ratio defines the decision-weighting risk-sensitivity metric.

The calculation of the probability premium then proceeds as follows.

\begin{itemize}
    \item \textit{Symmetric binary gamble around a stationary state.}
    For the gamble $\{x-\delta,1-p;\ x+\delta,p\}$, with $p=1/2$, the probability premium $\gamma$ is determined by
    \[
        u(x)
        =
        \omega^+(p+\gamma)u(x+\delta)
        +
        \omega^-(1-p-\gamma)u(x-\delta).
    \]
    This equation can be solved numerically. The search for $\gamma$ can be restricted by requiring
    \[
        \bigl(\omega^+(p+\gamma),\omega^-(1-p-\gamma)\bigr)
        \in
        \mathcal S(x,\delta,\delta,1).
    \]

    \item \textit{Non-symmetric binary gamble around a stationary state.}
    For the gamble $\{x-w\delta_1,1-p;\ x+w\delta_2,p\}$, where
    \[
        p=\frac{\delta_1}{\delta_1+\delta_2},
    \]
    the probability premium $\gamma$ is determined by
    \[
        u(x)
        =
        \omega^+(p+\gamma)u(x+w\delta_2)
        +
        \omega^-(1-p-\gamma)u(x-w\delta_1).
    \]
    Equivalently, the distorted decision-weight pair must satisfy
    \[
        \bigl(\omega^+(p+\gamma),\omega^-(1-p-\gamma)\bigr)
        \in
        \mathcal S(x,\delta_1,\delta_2,w).
    \]

    \item \textit{Nested binary gambles with a common probability distribution.}
    For the nested gambles $A$ and $B$ defined above, the probability premium $\gamma$ satisfies
    \[
    \omega^+(p+\gamma)u(x+\delta_4)
    +
    \omega^-(1-p-\gamma)u(x-\delta_1)
    =
    \omega^+(p+\gamma)u(x+\delta_2)
    +
    \omega^-(1-p-\gamma)u(x-\delta_3),
    \]
    where
    \[
        p=
        \frac{\delta_1-\delta_3}{(\delta_4-\delta_2)+(\delta_1-\delta_3)} .
    \]
    Equivalently,
    \[
        \frac{\omega^+(p+\gamma)}{\omega^-(1-p-\gamma)}
        =
        \frac{
        u(x-\delta_3)-u(x-\delta_1)
        }{
        u(x+\delta_4)-u(x+\delta_2)
        }
        =
        RS_{\rm gamble}^{\rm equal}(x,\delta_1,\delta_2,\delta_3,\delta_4).
    \]
\end{itemize}

As we can observe, when the decision setup and the utility function $u$ are fixed, the construction of the decision-weighting frontier $\mathcal S$ precedes the calculation of the probability premium $\gamma$. The frontier $\mathcal S$ isolates the effect of the utility function on the decision-weight threshold, whereas the probability weighting functions determine whether the distorted decision-weight pair satisfies this threshold. The probability premium combines both effects: the effect of the utility function and the effect of probability distortion.

\subsection{Generic Qualitative Analysis between Risk Premium, Probability Premium, Utility Premium and Probabilistic Risk Sensitivity}

After the presentation of the quantitative analysis about the probabilistic risk sensitivity and the utility premium in subsection \ref{section:Probabilistic_Utility_Premium} and probabilistic risk sensitivity and probability premium in subsection \ref{section:Probabilistic_Probability_Premium}, we are going to provide a qualitative analysis between the four different risk sensitivity metrics.

To start with, the risk premium and the utility premium are defined for gambles with an arbitrary number of outcomes. However, in the literature, the applications of utility premium in loss aversion are expressed to binary outcomes. In contrast, the proposed probabilistic metric and the probability premium are defined and applied to binary gambles in this work, but extending the framework to gambles with an arbitrary number of outcomes is left for future research.

Regarding the calculation of the risk sensitivity metrics, the risk premium analysis relies on first- and second-order Taylor approximations to compute ARA and RRA. In contrast, neither the utility premium, the probability premium nor the proposed probabilistic metric relies on approximation methods.

The risk premium framework describes the impact of the utility function with respect to its domain, while the utility premium describes the impact of the utility function with respect to its co-domain. On the other hand, the probability premium and the probabilistic risk sensitivity describe the effect of the utility function at the probability distribution. In addition, all four metrics provide information about the curvature of the utility function.

This divergence leads to the observation that the risk and probability premium are defined with goal to present the effect of the utility function on the construction of indifferent decision alternatives and at the same time give information about the decision strategy, which direction is less promoted in the literature. More precisely, the risk premium is related with the modification of the stationary state, while the probability premium refers to the tuning of the probability distribution of the gamble. The probabilistic risk sensitivity metric is created in order to provide a better insight to the decision strategy based on the distribution of the gamble providing also a rigorous graphical interpretation. Regarding the utility premium, only the sign gives information about the decision strategy of the agent.

Regarding the process of selection the parameters of the CPT utility function, it is non-rigorous to use different risk sensitivity metrics independently in order to proceed to this selection. For instance, it is inconsistent to use the risk premium framework (ARA/RRA) to define parameters across subdomains of the utility function while simultaneously employing the utility premium framework to define loss aversion. We have to underline that the risk premium framework, despite being approximated in practice at the calculation of ARA/RRA, describe the decision strategy quantitative under a decision setup between a stationary state and a gamble but there is a lack in literature about the use of risk premium to describe the risk behavior when there is interaction between the subdomains, which is the case of loss aversion. The ARA and RRA approaches cannot be applied in loss aversion definitions due to the lack of differentiability at the reference point. On the other hand, the utility premium framework describes quantitatively the intensity of acceptance or rejection and the definitions of loss aversion in the literature are deployed in inequation formulation under the fairness condition. It is important to underline that the probabilistic risk sensitivity or probability premium are more stringent risk-sensitivity metrics than the risk premium for the a decision setup between a stationary state and a gamble. Indeed, we are going to give an example based on the case of a non-symmetric gamble around a stationary state under fairness condition. More precisely, the risk and utility premium are defined as following:
\begin{subequations}
    \begin{equation*}
        u(x-\pi) = \frac{\delta_1}{\delta_1+\delta_2}u(x+w\delta_2) + \frac{\delta_2}{\delta_1+\delta_2}u(x-w\delta_1)
    \end{equation*}
    \begin{equation*}
        \mathcal{U}(x,\delta_1,\delta_2,w) = u(x) - u(x-\pi)
    \end{equation*}
\end{subequations}
We assume that the utility function satisfies the \emph{non-symmetric bet aversion} at the reference point $x=x_0$, $\pi > 0$. By taking the derivative of $\mathcal{U}$ with respect to $w$ at the reference point, we have the following by applying the chain rule of derivatives: 
\[
    \left. \frac{\partial \mathcal{U}(x,\delta_1,\delta_2,w)}{\partial w} \right|_{x=x_0} = - \left. \frac{\partial u(x-\pi)}{\partial (x - \pi)} \right|_{x=x_0} \left. \frac{\partial (x-\pi)}{\partial w} \right|_{x=x_0} = \left. \frac{\partial u(y)}{\partial (y)} \right|_{x=x_0} \left. \frac{\partial \pi}{\partial w} \right|_{x=x_0}
\]
We have to underline that as $x_0-\pi < x_0$ under the assumption of \emph{non-symmetric bet aversion}, the derivative $\left. \frac{\partial u(x-\pi)}{\partial (x - \pi)} \right|_{x=x_0}$ exists. Due to the strictly increasing monotonicity of the utility function, \textit{the derivative of the utility premium with respect to $w$ has the same sign with the derivative of the risk premium with respect to $w$}. Nevertheless, by taking into consideration the results of subsections \ref{section:Probabilistic_Utility_Premium} and \ref{section:Probabilistic_Probability_Premium}, we can conclude to the following lemma:
\begin{lemma}
    We assume the fundamental condition for loss aversion about the utility function, $\frac{u'(x_0^-)}{u'(x_0^+)} > 1$. If $\frac{\partial \mathcal{RS}(x_0,\delta_1,\delta_2,w)}{\partial w} > 0$, then $\frac{\partial \mathcal{U}(x_0,\delta_1,\delta_2,w)}{\partial w} > 0 \Leftrightarrow \left. \frac{\partial \pi}{\partial w} \right|_{x=x_0} > 0$. On the other hand, if $\left. \frac{\partial \pi}{\partial w} \right|_{x=x_0} > 0 \Leftrightarrow \frac{\partial \mathcal{U}(x_0,\delta_1,\delta_2,w)}{\partial w} > 0$, then it is not, in general, implied that $\frac{\partial \mathcal{RS}(x_0,\delta_1,\delta_2,w)}{\partial w} > 0$.
\end{lemma}

By contrast, the proposed probabilistic metric identifies threshold probability distributions that determine the preference of one decision entity over another giving information for both fair and non-fair gambles. These entities may be stationary states or gambles with equal or unequal probability distributions. In this work, we focus on binary gambles as the elementary case of the proposed framework; however, the method can be generalized to $N$-outcome gambles. 

To sum up, probabilistic risk sensitivity provides a complementary and directly interpretable tool for describing risk behavior as a function of the probability distribution. It should be used {primary and at the same time coherently} with, rather than independently from, risk-premium and utility-premium measures when determining utility-function parameters.

\section{Conclusion}\label{Section5}

This work uses a generalized framework for interpreting risk sensitivity, which is visualized through utility curvature, and loss aversion through the lens of binary gambles, offering a structured perspective on subjective risk preferences beyond the scope of traditional models based on the utility premium. By reframing key behavioral concepts, such as symmetric and non-symmetric bet aversion along with their increasing definitions and the general cases of binary gambles with equal or unequal probability distributions, within a binary gamble framework, we not only clarified the foundational definitions but also addressed theoretical ambiguities in the existing literature. In addition, we compare different risk-sensitivity metrics and clarify how they should be used when setting value-function parameters. 

Our analysis demonstrated that binary gambles serve not merely as illustrative tools, but as a unifying methodological foundation capable of recovering and extending canonical concepts in CPT. 

The proposed generalized framework opens several promising research directions. First, it calls for empirical validation across diverse decision-making contexts, such as intertemporal choice, environmental risk, or AI-human interaction systems. Second, it offers potential applications in algorithmic design for decision support systems that must account for nuanced and context-dependent human risk preferences. Finally, it raises foundational questions about the normative versus descriptive role of utility in behavioral modeling, especially as decision systems evolve to be more adaptive and personalized.

In conclusion, by enriching the theoretical foundations of CPT with structurally grounded and behaviorally consistent extensions, this work advances the ongoing development of risk-aware decision theory, offering a framework that more faithfully reflects the complexity of human judgment and choice under uncertainty.

\section*{Acknowledgment}
This work is supported by the European Research Council (ERC) under the EU’s Horizon 2020 research and innovation programme (Grant agreement No. 101003431). S. Vaidanis is partially supported by the Onassis Foundation (Scholarship ID: F ZU 076-1/2024-2025). 

\begin{appendices}

\section{Proofs of the Theorems}
In this section, we present the proofs of the theorems referenced in the main text. Before proceeding, we introduce several useful propositions that will serve as foundational tools for the subsequent proofs.

\subsection{Proof of Proposition \ref{concave_direct}}
\begin{proposition} \label{concave_direct}
If the utility function $u$ is strictly increasing and strictly concave over a subdomain, then for any $x$ in that subdomain and any $\delta_1, \delta_2 > 0$, it holds that
\[
\mathbb{P}_{\text{gain}}(x,\delta_1,\delta_2) > \frac{\delta_1}{\delta_1 + \delta_2}.
\]
\end{proposition}

\begin{proof}
We begin with the expression for the gain threshold probability:
\[
\mathbb{P}_{\text{gain}}(x,\delta_1,\delta_2) = \frac{u(x) - u(x - \delta_1)}{u(x + \delta_2) - u(x - \delta_1)}.
\]
We want to show that $\frac{u(x) - u(x - \delta_1)}{u(x + \delta_2) - u(x - \delta_1)} > \frac{\delta_1}{\delta_1 + \delta_2}$.
This is equivalent to $\delta_2 \cdot \big(u(x) - u(x - \delta_1)\big) > \delta_1 \cdot \big(u(x + \delta_2) - u(x)\big)$,
which can be rewritten as $\frac{u(x) - u(x - \delta_1)}{\delta_1} > \frac{u(x + \delta_2) - u(x)}{\delta_2}$.

Applying now the Mean Value Theorem to the intervals $[x - \delta_1, x]$ and $[x, x + \delta_2]$, there exist points $\xi \in (x - \delta_1, x)$ and $\rho \in (x, x + \delta_2)$ such that:
\[
u'(\xi) = \frac{u(x) - u(x - \delta_1)}{\delta_1}, \quad u'(\rho) = \frac{u(x + \delta_2) - u(x)}{\delta_2}.
\]
Since $u$ is strictly concave, its derivative is strictly decreasing. Hence, $\xi < \rho$ implies $u'(\xi) > u'(\rho)$, and therefore $\frac{u(x) - u(x - \delta_1)}{\delta_1} > \frac{u(x + \delta_2) - u(x)}{\delta_2}$.
\end{proof}

\subsection{Proof of Proposition \ref{convex_direct}}
\begin{proposition} \label{convex_direct}
If the utility function $u$ is strictly increasing and strictly convex over a subdomain, then for any $x$ in that subdomain and any $\delta_1, \delta_2 > 0$, it holds that
\[
\mathbb{P}_{\text{gain}}(x,\delta_1,\delta_2) < \frac{\delta_1}{\delta_1 + \delta_2}.
\]
\end{proposition}

\begin{proof}
Recall that the gain probability threshold is defined as:
\[
\mathbb{P}_{\text{gain}}(x,\delta_1,\delta_2) = \frac{u(x) - u(x - \delta_1)}{u(x + \delta_2) - u(x - \delta_1)}.
\]
We aim to show that $\frac{u(x) - u(x - \delta_1)}{u(x + \delta_2) - u(x - \delta_1)} < \frac{\delta_1}{\delta_1 + \delta_2}$. 
Cross-multiplying gives the equivalent inequality:
\[
\delta_2 \cdot \big(u(x) - u(x - \delta_1)\big) < \delta_1 \cdot \big(u(x + \delta_2) - u(x)\big).
\]
Dividing both sides by $\delta_1 \delta_2$ yields $\frac{u(x) - u(x - \delta_1)}{\delta_1} < \frac{u(x + \delta_2) - u(x)}{\delta_2}$.

Now, apply the Mean Value Theorem to each interval:
\begin{itemize}
    \item On $[x - \delta_1, x]$, there exists $\xi \in (x - \delta_1, x)$ such that $u'(\xi) = \frac{u(x) - u(x - \delta_1)}{\delta_1}$.
    \item On $[x, x + \delta_2]$, there exists $\rho \in (x, x + \delta_2)$ such that $u'(\rho) = \frac{u(x + \delta_2) - u(x)}{\delta_2}$.
\end{itemize}

Because $u$ is strictly convex, its derivative $u'$ is strictly increasing. Since $\xi < \rho$, it follows that $u'(\xi) < u'(\rho)$, and thus, $\frac{u(x) - u(x - \delta_1)}{\delta_1} < \frac{u(x + \delta_2) - u(x)}{\delta_2}$.
\end{proof}

\subsection{Proof of Proposition \ref{linear_direct}}
\begin{proposition} \label{linear_direct}
If the utility function $u$ is strictly increasing and linear over a subdomain, then for any $x$ in that subdomain and any $\delta_1, \delta_2 > 0$, it holds that
\[
\mathbb{P}_{\text{gain}}(x,\delta_1,\delta_2) = \frac{\delta_1}{\delta_1 + \delta_2}.
\]
\end{proposition}

\begin{proof}
Assume that $u$ is linear on the interval $[x - \delta_1, x + \delta_2]$, so it can be written as $u(z) = \alpha z + \beta$ for some constants $\alpha > 0$ and $\beta \in \mathbb{R}$.

Then:
\[
\begin{aligned}
u(x - \delta_1) &= \alpha(x - \delta_1) + \beta, \\
u(x) &= \alpha x + \beta, \\
u(x + \delta_2) &= \alpha(x + \delta_2) + \beta.
\end{aligned}
\]

Substituting into the definition of $\mathbb{P}_{\text{gain}}$: $\mathbb{P}_{\text{gain}}(x,\delta_1,\delta_2) = \frac{u(x) - u(x - \delta_1)}{u(x + \delta_2) - u(x - \delta_1)}$.

Using the expressions above:
\[
\begin{aligned}
u(x) - u(x - \delta_1) &= \alpha x + \beta - \left[\alpha(x - \delta_1) + \beta\right] = \alpha \delta_1, \\
u(x + \delta_2) - u(x - \delta_1) &= \alpha(x + \delta_2 - x + \delta_1) = \alpha(\delta_1 + \delta_2).
\end{aligned}
\]

Therefore: $\mathbb{P}_{\text{gain}}(x,\delta_1,\delta_2) = \frac{\alpha \delta_1}{\alpha(\delta_1 + \delta_2)} = \frac{\delta_1}{\delta_1 + \delta_2}$.
\end{proof}

\subsection{Proof of Lemma \ref{EUT_symmetric}}

To start with, we will prove the \emph{probabilistic symmetric bet aversion}: $\mathcal{RS}_{stat}(x,\delta) > 1 \Leftrightarrow \frac{\frac{u(x) - u(x-\delta)}{x-(x-\delta)}}{\frac{u(x+\delta)-u(x)}{(x+\delta)-x}}>1$ which holds true because by applying the Mean Value Theorem at the intervals $[x-\delta,x]$ and $[x,x+\delta]$ we have that $\exists p \in (x-\delta,x)$ and $\exists \xi \in (x,x+\delta)$ such that $u'(p) = \frac{u(x) - u(x-\delta)}{x-(x-\delta)},\; u'(\xi)=\frac{u(x+\delta)-u(x)}{(x+\delta)-x}$ and due to the curvature of u, we have that $u'(p)>u'(\xi)$.

Next, we will prove the \emph{probabilistic increasing symmetric bet aversion}: $\frac{\partial \mathcal{RS}_{stat}(x,\delta)}{\partial \delta} > 0 \Leftrightarrow \frac{u'(x-\delta)}{u'(x+\delta)} > \frac{u(x) - u(x-\delta)}{u(x+\delta)-u(x)} \Leftrightarrow \frac{u'(x-\delta)}{u'(x+\delta)} > \frac{\frac{u(x) - u(x-\delta)}{x-(x-\delta)}}{\frac{u(x+\delta)-u(x)}{(x+\delta)-x}}$ which holds true because by applying the Mean Value Theorem at the intervals $[x-\delta,x]$ and $[x,x+\delta]$ we have that $\exists p \in (x-\delta,x)$ and $\exists \xi \in (x,x+\delta)$ such that $u'(p) = \frac{u(x) - u(x-\delta)}{x-(x-\delta)},\; u'(\xi)=\frac{u(x+\delta)-u(x)}{(x+\delta)-x}$ and due to the monotonicity and the curvature of u, we have that $u'(x-\delta)>u'(p), \; u'(x+\delta)<u'(\xi)$.

\subsection{Proof of Lemma \ref{EUT_non_symmetric}}

To start with, we will prove the \emph{probabilistic non-symmetric bet aversion}: $\mathcal{RS}_{stat}(x,\delta_1,\delta_2,w) > \frac{\delta_1}{\delta_2} \Leftrightarrow \frac{\frac{u(x) - u(x-w\cdot\delta_1)}{x-(x-w\cdot\delta_1)}}{\frac{u(x+w\cdot\delta_2)-u(x)}{(x+w\cdot\delta_2)-x}}>1 \Leftrightarrow $ which holds true because by applying the Mean Value Theorem at the intervals $[x-w\cdot\delta_1,x]$ and $[x,x+w\cdot\delta_2]$ we have that $\exists p \in (x-w\cdot\delta_1,x)$ and $\exists \xi \in (x,x+w\cdot\delta_2)$ such that $u'(p) = \frac{u(x) - u(x-w\cdot\delta_1)}{x-(x-w\cdot\delta_1)},\; u'(\xi)=\frac{u(x+w\cdot\delta_2)-u(x)}{(x+w\cdot\delta_2)-x}$ and due to the curvature of u, we have that $u'(p)>u'(\xi)$.

Next, we will prove the \emph{probabilistic increasing symmetric bet aversion}: $\frac{\partial \mathcal{RS}_{stat}(x,\delta_1,\delta_2,w)}{\partial w} > 0 \Leftrightarrow \frac{u'(x-w\cdot\delta_1)}{u'(x+w\cdot\delta_2)} > \frac{\delta_2}{\delta_1} \cdot \frac{u(x) - u(x-w\cdot\delta_1)}{u(x+w\cdot\delta_2)-u(x)} \Leftrightarrow \frac{u'(x-w\cdot\delta_1)}{u'(x+w\cdot\delta_2)} > \frac{\frac{u(x) - u(x-w\cdot\delta_1)}{x-(x-w\cdot\delta_1)}}{\frac{u(x+w\cdot\delta_2)-u(x)}{(x+w\cdot\delta_2)-x}}$ which holds true because by applying the Mean Value Theorem at the intervals $[x-w\cdot\delta_1,x]$ and $[x,x+w\cdot\delta_2]$ we have that $\exists p \in (x-w\cdot\delta_1,x)$ and $\exists \xi \in (x,x+w\cdot\delta_2)$ such that $u'(p) = \frac{u(x) - u(x-w\cdot\delta_1)}{x-(x-w\cdot\delta_1)},\; u'(\xi)=\frac{u(x+w\cdot\delta_2)-u(x)}{(x+w\cdot\delta_2)-x}$ and due to the monotonicity and the curvature of u, we have that $u'(x-w\cdot\delta_1)>u'(p), \; u'(x+w\cdot\delta_2)<u'(\xi)$.

\subsection{Proof of Lemma \ref{EUT_Neilson}}

Regarding the \emph{Neilson' strong aversion}, by applying the Mean Value Theorem at the intervals $[x-\delta_1,x-\delta_3]$ and $[x+\delta_2,x+\delta_4]$ we have that $\exists p \in (x-\delta_1,x-\delta_3)$ and $\exists \xi \in (x+\delta_2,x+\delta_4)$ such that $u'(p) = \frac{u(x-\delta_3) - u(x-\delta_1)}{(x-\delta_3)-(x-\delta_1)},\; u'(\xi)=\frac{u(x+\delta_4)-u(x+\delta_2)}{(x+\delta_4)-(x+\delta_2)}$ and due to the curvature of u, we have that $u'(p)>u'(\xi)$.

Regarding the \emph{Probabilistic increasing symmetric nested bet aversion with equal probability distribution}, we have the following: $\frac{\partial \mathcal{RS}^{equal}_{gamble}(x,\delta_1,\delta_2,\delta_2,\delta_1)}{\partial \delta_1} > 0 \Leftrightarrow \frac{u'(x-\delta_1)}{u'(x+\delta_1)} > \frac{u(x-\delta_2) - u(x-\delta_1)}{u(x+\delta_1)-u(x+\delta_2)} \Leftrightarrow \frac{u'(x-\delta_1)}{u'(x+\delta_1)} > \frac{\frac{u(x-\delta_2) - u(x-\delta_1)}{(x-\delta_2)-(x-\delta_1)}}{\frac{u(x+\delta_1)-u(x+\delta_2)}{(x+\delta_1)-(x+\delta_2)}}$ which holds true because by applying the Mean Value Theorem at the intervals $[x-\delta_1,x-\delta_2]$ and $[x+\delta_2,x+\delta_1]$ we have that $\exists p \in (x-\delta_1,x-\delta_2)$ and $\exists \xi \in (x+\delta_2,x+\delta_1)$ such that $u'(p) = \frac{u(x-\delta_2) - u(x-\delta_1)}{((x-\delta_2)-(x-\delta_1))},\; u'(\xi)=\frac{u(x+\delta_1)-u(x+\delta_2)}{(x+\delta_1)-(x+\delta_2)}$ and due to the monotonicity and the curvature of u, we have that $u'(x-\delta_1)>u'(p), \; u'(x+\delta_1)<u'(\xi)$.

\subsection{Proof of Theorem \ref{concave_equivalent}}
Let $a=x-\delta_1$ and $b=x+\delta_2$, so that $a<x<b$ and
\[
    x=\frac{\delta_2}{\delta_1+\delta_2}a+\frac{\delta_1}{\delta_1+\delta_2}b.
\]
Since $u$ is strictly increasing, $u(b)-u(a)>0$. The inequality
\[
\mathbb{P}_{\text{gain}}(x,\delta_1,\delta_2)>\frac{\delta_1}{\delta_1+\delta_2}
\]
is equivalent to
\[
    u(x)>\frac{\delta_2}{\delta_1+\delta_2}u(a)+\frac{\delta_1}{\delta_1+\delta_2}u(b),
\]
which is exactly the strict concavity inequality at the interior point $x$. Because every triple $(a,x,b)$ with $a<x<b$ can be represented by $\delta_1=x-a$ and $\delta_2=b-x$, the condition holds for all admissible triples if and only if $u$ is strictly concave on $I$.

\subsection{Proof of Theorem \ref{convex_equivalent}}
With $a=x-\delta_1$ and $b=x+\delta_2$, the inequality
\[
\mathbb{P}_{\text{gain}}(x,\delta_1,\delta_2)<\frac{\delta_1}{\delta_1+\delta_2}
\]
is equivalent to
\[
    u(x)<\frac{\delta_2}{\delta_1+\delta_2}u(a)+\frac{\delta_1}{\delta_1+\delta_2}u(b).
\]
This is precisely the strict convexity inequality at each interior point of each interval $[a,b]\subset I$. Since all such interior points are generated by admissible $x,\delta_1,\delta_2$, the stated condition is equivalent to strict convexity on $I$.

\subsection{Proof of Theorem \ref{linear_equivalent}}
Again set $a=x-\delta_1$ and $b=x+\delta_2$. The equality
\[
\mathbb{P}_{\text{gain}}(x,\delta_1,\delta_2)=\frac{\delta_1}{\delta_1+\delta_2}
\]
is equivalent to
\[
    u(x)=\frac{\delta_2}{\delta_1+\delta_2}u(a)+\frac{\delta_1}{\delta_1+\delta_2}u(b).
\]
Thus $u$ coincides with its chord at every interior point of every interval contained in $I$. This chord-linearity property is equivalent to affinity on $I$. Conversely, any affine strictly increasing function satisfies the equality immediately by substitution.

\end{appendices}

\bibliographystyle{IEEEtran}
\bibliography{my_bibliography}

\end{document}